\documentclass[10pt]{article}
\usepackage{amsmath, amsfonts, amsthm, amssymb,verbatim}
\usepackage[a4paper,dvips,nohead]{geometry}
\usepackage{xspace}
\usepackage{subfigure}
\usepackage{graphicx}

\newtheorem{theorem}{Theorem}[section] 
\newtheorem{definition}[theorem]{Definition}
\newtheorem{lemma}[theorem]{Lemma}

\newtheorem{proposition}[theorem]{Proposition}
\newtheorem{example}[theorem]{Example}

\numberwithin{equation}{section}
\usepackage[ps2pdf, bookmarks=true,
            bookmarksopen=true,colorlinks=false]{hyperref}

\newcommand \bz {\begin{itemize}}
\newcommand \ez {\end{itemize}}
\newcommand \ben {\begin{enumerate}}
\newcommand\een {\end{enumerate}} 
\newcommand \N {\mathbb N}
\newcommand \R {\mathbb R} 
\newcommand{\Eqref}[1]{\eqref{#1}}
\newcommand{\Eqsref}[1]{\eqref{#1}}
\newcommand{\Sectionref}[1]{Section~\ref{#1}}  

\newcommand{\Lemref}[1]{Lemma~\ref{#1}}
\newcommand{\Propref}[1]{Proposition~\ref{#1}}
\newcommand{\Theoremref}[1]{Theorem~\ref{#1}}

\newcommand{\Figref}[1]{Figure~\ref{#1}}
\newcommand \auth {} 
\newcommand \jou {\textit}

\newcommand \muh M 
\newcommand \vv u

\newcommand \RR 		{\mathbb{R}}  
\newcommand \del 	\partial
\newcommand \eps 	\epsilon
\newcommand \lam 	\lambda 
\newcommand \be 		{\begin{equation}}
\newcommand\ee 		{\end{equation}}

\newcommand \alphat {\widetilde \alpha} 
\newcommand \what {\widehat w}
\newcommand \Xt {\widetilde X} 
\newcommand \Et {\widetilde E}
\newcommand \Pt {\widetilde P} 
\newcommand \Qt {\widetilde Q}
\begin{document}

\title{Second-order hyperbolic Fuchsian systems
\\
 and applications} 

\author{Florian Beyer$^1$ and Philippe G. LeFloch$^2$ 
}

\date{October 2010}

\maketitle 

\footnotetext[1]{
Department of Mathematics and Statistics, University of Otago, P.O.~Box 56, Dunedin 9054, New Zealand, E-mail: {fbeyer@maths.otago.ac.nz}.
\newline
$^2$Laboratoire Jacques-Louis Lions \& Centre National de la Recherche Scientifique, 
Universit\'e Pierre et Marie Curie (Paris 6), 
4 Place Jussieu, 75252 Paris, France. E-mail: {pgLeFloch@gmail.com.}
\newline 
Published in: {\tt Class. Quantum Grav. 27  (2010), 245012.}
\hfill
}

\begin{abstract}
  We introduce a new class of singular partial differential equations, referred to as 
  the
  {\sl second-order hyperbolic Fuchsian systems,} and we investigate the
  associated initial value problem when data are imposed on the
  singularity.  First, we establish a general existence theory of solutions
  with asymptotic behavior prescribed on the singularity, which
   relies on
  a new approximation scheme, suitable also for numerical purposes.
  Second, 
this theory is applied to the (vacuum) Einstein equations for
  Gowdy spacetimes, and allows us to recover, by more direct arguments, 
  {\sl well-posedness results} established earlier by Rendall and
  collaborators. Another main contribution in this paper is the proposed
  approximation scheme, which we refer to as the {\sl Fuchsian numerical
  algorithm} and is shown to provide highly accurate numerical
  approximations to the singular initial value problem. For the class
  of Gowdy spacetimes, the numerical experiments
  presented here show the interest and efficiency of the proposed
  method and demonstrate the existence of a class of Gowdy spacetimes containing a
  smooth, incomplete, and non-compact Cauchy horizon.
\end{abstract}

\tableofcontents


\section{Introduction}
 
In general relativity and, more generally, in the theory of  partial differential equations, 
singular solutions
play a central role in driving the development of the theory 
and interpretation of the results 
-- from the discovery of the
Schwarzschild solution to the Penrose and Hawking's celebrated singularity theorems
(cf.~for instance \cite{hawking}), 
as well as in the modern efforts to understand the
cosmic censorship and BKL (Belinsky�-Khalatnikov�-Lifshitz) conjectures 
\cite{Andersson,Ringstrom7}.

In fundamental and pioneer work, 
Choquet-Bruhat \cite{Choquet52}  established 
that the initial value problem associated with Einstein's field equations is
well-posed (in suitable Sobolev spaces) 
and, later together with Geroch \cite{Choquet69},
that for each choice of initial data consistent with
the constraints there exists a unique maximal globally hyperbolic
development. This theory allows to define a unique correspondence, between
the space
of solutions of the Einstein equations 
and the space of initial data set on a given 
(spacelike) hypersurface. In principle, solutions could be extended until 
a singularity forms and this theory provides the basis 
to tackle outstanding questions about the long-time behavior of solutions, 
in both future and past directions. Consider for instance the situation that 
singularities develop {\sl in the past} of the given hypersurface, and let us refer to the construction
based on \cite{Choquet52,Choquet69} 
as the ``backward approach''. 
In practice, however, revealing information about singular solutions requires 
additional techniques of analysis going much beyond \cite{Choquet52,Choquet69}. 
On one hand, there is no a-priori
information whether a given choice of initial data evolves into a
``singularity'' at all. On the other hand,
if a singularity arises, it is not a-priori clear at ``which location'' it will 
occur, nor what kind of ``singular solution'' it will be. 
This approach requires a global-in-time control of solutions which is  
hopeless in many cases due to the complexity of the nonlinear field equations, the
freedom of choice of gauge etc. In fact, the ``backward approach'' 
has so far led to successful conclusions only under strong symmetry assumptions; 
cf.~Ringstr\"om \cite{Ringstrom7} (and the references therein).

The so-called Fuchsian method, which we refer to here as a
``forward'' approach, provides an alternative to the above (backward) approach. It has the
advantage of being simpler to deal with in practice (both 
analytically and
numerically) in the class of problems it does apply, 
but has the disadvantage that it does not allow -- for essential reasons -- to handle the 
full solution space and the conclusions are stated in terms of ``generic data''. 
Nevertheless it allows to study particular
classes of singular solutions. The main idea is to give data ``on the
singularity'', that is, 
to prescribe the leading-order behavior
of solutions in a neighborhood of the singular time and, then, to evolve forward
in time away from the singularity -- in contrast to prescribing data
on a Cauchy surface and evolving towards the singularity. 
(We will make this idea precise below.) 
Our main objective in the present 
paper is precisely 
to further 
study this \textbf{singular initial value problem} (corresponding to the forward direction).

The Fuchsian method was introduced to general relativity by
Kichenassamy and Rendall \cite{KichenassamyRendall} who
covered a class of singular equations, the so-called ``Fuchsian partial
  differential equations'', 
  while
  establishing that Einstein's field equations under Gowdy
symmetry are included in this class. 
 Under suitable analyticity conditions on the solutions they
proved the well-posedness of the singular initial value
problem (without imposing any hyperbolicity condition).
 Later, Rendall \cite{Rendall00} took the hyperbolicity
property into account and generalized the theory
to the smooth solutions to the Gowdy equations. 
For further results about Fuchsian equations we refer to 
\cite{Kichenassamy,ChoquetIsenberg,Choquet08,Choquet09}; especially,
in \cite{Choquet08}, a generalization of the standard Fuchsian theory was recently introduced.

Our motivation in this
paper now is threefold. First, it is of general interest to find a
reliable numerical scheme for the singular initial value problem of
hyperbolic Fuchsian equations, which 
would allow us to construct
solutions to Einstein's field equations with prescribed singular
behavior. This was indeed the motivation of Amorim, Bernardi, and LeFloch 
in  
\cite{ABL}. 
However, 
the approximation scheme suggested by the theoretical works
above
is neither suitable nor natural for the numerical
treatment, and it is not easy to obtain error estimates which 
are necessary in practice to judge the quality of the numerical solutions. 
In this paper, we address this drawback by introducing a new
approximation scheme. 

Secondly, we want to perform these studies ---both theoretically and numerically---
as economically as possible. Importantly,
the Einstein equations for Gowdy spacetimes are naturally expressed in a 
second-order form.  This motivates us to develop here
a theory of {\bf second-order Fuchsian equations}
which is advantageous as it saves us from the hassle of turning
the system into a first-order form first, as required by the classical 
Fuchsian theory. 
In fact, there is no unique way of
turning a second-order system into a first-order form, 
and this issue can be
particularly problematic when the solutions are
expected to be singular. 
It has also been recognized in the numerical literature 
that the direct discretization in second-order form leads to more accurate results 
\cite{Kreiss}. 

Our 
third motivation for this paper is that the existing works make
no statement about how to guess the leading-order part.
This is a delicate issue and we prove here that, for a large class of systems
and in  way 
 compatible with the (already mentioned) BKL conjecture, it is 
 possible to make a ``canonical guess'', as we explain below. 

These three main issues are addressed in the present paper, which is 
organized as follows. 

Sections~\ref{sec:2ndorderfuchsiantheory} and
\ref{sec:existence} are the main theoretical part.
In \Sectionref{sec:2ndorderfuchsianterminology}, we introduce 
the
class of equations of interest in this work. 
We focus on a class of hyperbolic Fuchsian equations, the 
\textbf{second-order hyperbolic Fuchsian equations}, which
are systems
of semi-linear wave equations including certain singular terms.  
This section includes a rigorous definition of the singular
initial value problem for this class of
equations. \Sectionref{sec:canonicallot} is devoted to determine the
leading-order behavior of solutions and hence to make a ``canonical
guess''. 
We study the case when the principal
part of the equation dominates over the source-term at $t=0$ in a
certain sense and derive a \textbf{canonical two-term expansion}. The
reasoning is based on suitable heuristics compatible
with the BKL conjecture. (Later in this paper we derive
precise conditions under which this heuristics is justified.)

The singular initial value problem is discussed rigorously in \Sectionref{sec:existence}. We introduce our new
approximation scheme which, both, yields a simple and direct proof
of existence of solutions to
 the singular initial value problem, and a natural numerical algorithm
 (introduced
in \Sectionref{sec:numerical_scheme})
for which practical error
estimates can be obtained. 
The main idea is to approximate the solution of the
singular initial value problem by a sequence of solutions to the
standard (regular) initial value problem. 

Finally, in \Sectionref{sec:Gowdy}, 
we turn our attention to the class of Gowdy
spacetimes satisfying Einstein's field equations which is an application of particular interest, 
and we demonstrate the practical use of 
our theory. First, we recapitulate the standard heuristic arguments
for Gowdy spacetimes and demonstrate that these are consistent with the
heuristics introduced earlier (with some important subtleties discussed below). 
In \Sectionref{sec:gowdyequations2ndhypFuchs}, the now classical 
results by Rendall \cite{Rendall00} are recovered while 
our new approach shed some new
light on the Gowdy equations. 
In \Sectionref{sec:numericalsolGowdy} we present numerical
solutions to the singular initial value problem associated with the Gowdy
equations. After some test cases, we numerically construct solutions with
incomplete, non-compact Cauchy horizons.


\section{Second-order hyperbolic Fuchsian systems} 
\label{sec:2ndorderfuchsiantheory}

\subsection{A class of singular equations}
\label{sec:2ndorderfuchsianterminology}

\begin{definition}
A second-order hyperbolic Fuchsian system is a set of partial
differential equations of the form
\be
    \label{eq:secondorderFuchsianHyp}
      D^2 v+2A\, D v+B\, v 
      -t^2K^2\partial_x^2 v
      = f[v],
\ee
in which the function $v:(0,\delta]\times U\rightarrow \R^n$ is the
main unknown (defined for some $\delta>0$ and some interval $U$),
while the coefficients $A=A(x)$, $B=B(x)$, $K=K(t,x)$ are diagonal
$n\times n$ matrix-valued maps and are smooth in $x \in U$ and $t$ in
the half-open interval $(0,\delta]$, and $f=f[v](t, x)$ is an
$n$-vector-valued map of the following form
\begin{equation*} 
  f[v](t,x):=f\Big(t, x, v(t,x), Dv(t,x), t K(t,x) \del_x v(t,x)\Big).
\end{equation*}
\end{definition}

We assume that the time variable $t$ satisfies $t>0$ and use the
operator
$
D:=t\del_t
$ 
to write the equations. The equation is henceforth assumed to be
singular at $t=0$. The assumption that $U$ is a one-dimensional domain
makes the presentation simpler, but most results below remain valid
for arbitrary spatial dimensions.  For definiteness and without much
loss of generality, we assume throughout this paper that all functions
under consideration are periodic in the spatial variable $x$ and that
$U$ is the periodicity domain. All data and solutions are extended by
periodicity outside the interval $U$.  Moreover, we assume that the
coefficients $A$ and $B$ do not depend on $t$, see below.  We denote
the eigenvalues of $A$ and $B$ by $a^{(1)},\ldots,a^{(n)}$ and
$b^{(1)},\ldots,b^{(n)}$, respectively. When it is not necessary to
specify the superscripts, we just write $a, b$ to denote any
eigenvalues of $A,B$.  With this convention, we introduce:
\begin{equation}
  \label{eq:deflambda2}
  \lambda_{1}:=a+\sqrt{a^2-b},\quad \lambda_{2}:=a-\sqrt{a^2-b}.  
\end{equation}
It will turn out that these coefficients, which might be complex in
general, are important to describe the expected behavior at $t=0$ of
general solutions to \Eqref{eq:secondorderFuchsianHyp}.  Further
restrictions on the coefficients and on the right-hand side will be
imposed and discussed in the course of our investigation.

After a suitable reduction to a first-order form, our choice of
singular equations falls into the class of hyperbolic Fuchsian
equations
\cite{KichenassamyRendall,Rendall00,Kichenassamy,Choquet08}. We make
this particular choice of equations here for the following
reasons. First we restrict to hyperbolic equations because this is the
case of interest and allows us to control solutions in much greater
detail than without this assumption. Second, the equations are kept in
second-order form here because it is economic and efficient to do this
for the applications we have in mind both analytically and
numerically. Third, the particular class of equations allows us to
simplify the presentation of the general results obtained in this
paper.  However, all results presented here can be generalized to a
general class of symmetric hyperbolic Fuchsian equations -- compare to
the discussion in Rendall \cite{Rendall00} -- for arbitrary spatial
dimensions. Even beyond this it is possible to generalize the theory
to equations with time-dependent coefficients $A$ and $B$. Also the
restriction to spatial periodicity is not essential because these
equations obey the domain of dependence property just as usual
non-singular hyperbolic equations under suitable assumptions on $K$,
see below.

The eigenvalues of the matrix $K$ are denoted by $k^{(i)}$ and, in the
scalar case (or when there is no need to specify the index), we
simply write $k$. These quantities are interpreted as
characteristic speeds. Throughout this section, we assume that 
they have the form 
\begin{equation}
  \label{eq:behaviorofk}
  \begin{split}
    &k^{(i)}(t,x)=t^{\beta^{(i)}(x)}\nu^{(i)}(t,x),\\
    &\text{with }\beta^{(i)}:U\rightarrow (-1,\infty),\,
    \nu^{(i)}:[0,\delta]\times U\rightarrow (0,\infty)\text{ smooth functions.}
  \end{split}
\end{equation}
In particular, we assume that each derivative of $\nu^{(i)}$ has a
unique finite limit at $t=0$ for each $x\in U$.  Note that we allow
for the characteristic speeds to diverge at $t=0$. At a first glance,
this appears to conflict with the standard finite domain of dependence
property of hyperbolic equations. A closer look at the requirement
$\beta(x)>-1$, however, indicates that the characteristic curves are
{\sl integrable} at $t=0$ and, hence that the {\sl finite} domain of
dependence property is preserved under our assumptions.

The operator associated with the \textbf{principal part} of the system
is
\begin{equation}
  \label{eq:defLPDE}
  L:=D^2+2A \, D +B-t^2K^2\partial_x^2=:\tilde L-t^2K^2\partial_x^2.
\end{equation}
This is a linear wave operator for $t>0$ and, indeed,
\Eqref{eq:secondorderFuchsianHyp} is hyperbolic for all $t>0$. Later
on we will construct solutions where the first three terms $\tilde L$
of the principal part are of the same order at $t=0$ and ``dominant'',
while the \textbf{source-term} of the equation as well as the second spatial
derivative term are assumed to be of higher-order in $t$ at $t=0$ and
hence ``negligible''. Note that at this level generality, there is
some freedom in bringing terms from the principal part to the
right-hand side of the equation, and absorbing them into the
source-function $f$ (or vice-versa). This freedom has several
(interesting) consequences: roughly speaking, some normalization will
be necessary later, yet at this stage, we do not fix the behavior of
$f$ at $t=0$.

\subsection{Singular initial value problem} 

Consider any second-order hyperbolic Fuchsian system with coefficients
$a,b, \lambda_1, \lambda_2$, satisfying \Eqref{eq:deflambda2}.  To
simplify the presentation, we restrict attention to scalar equations
($n=1$) and shortly comment on the general case in the course of the
discussion.

Fix some integers $l,m \geq 0$ and constants $\alpha,\delta>0$.  For $w\in
C^{l}((0,\delta],H^m(U))$, we define the norm 
\begin{equation*} 
\|w\|_{\delta,\alpha,l,m}:=\sup_{0<t\le \delta} \left(
  \sum_{p=0}^l\sum_{q=0}^m\int_U t^{2(\Re\lambda_2(x)-\alpha)}
  \,\bigl|\del_x^q D^pw(t,x)\bigr|^2 \, dx \right)^{1/2},
\end{equation*}
and denote by $X_{\delta,\alpha,l,m}$ the space of all such functions
with finite norm $\|w\|_{\delta,\alpha,l,m}<\infty$.  Throughout,
$H^m(U)$ denotes the standard Sobolev space and we recall that all
functions are periodic in the variable $x$ with $U$ being a
periodicity domain.  To cover a system of $n\ge 1$ second-order
Fuchsian equations, the norm above is defined by summing over all
vector components with different exponents used for different
components. Recall that each equation in the system will have a
different root function $\lambda_2$.  We allow that
$\alpha=(\alpha^{(1)},\ldots,\alpha^{(n)})$ is a vector of different
positive constants for each equation. The constant $\delta$, however,
is assumed to be common for all equations in the system. With this
modification, all results in the present section remain valid for
systems of equations. We comment later that in fact, $\alpha$ is not
required to be a constant, but in most of the following results it
will be treated like a constant for simplicity. The reason to include
the quantity $\lambda_2$ into the definition of the norms is motivated
by the canonical choice of the leading-order term introduced later.
Throughout it is assumed that $\Re\lambda_2$ is continuous and it is
then easy to check that
$(X_{\delta,\alpha,l,m},\|\cdot\|_{\delta,\alpha,l,m})$ is a Banach
space.

For the discussion of hyperbolic equations, it makes sense to also
introduce the following Banach spaces. For each non-negative integer
$l$ and real numbers $\delta,\alpha>0$, we define
$X_{\delta,\alpha,l}:=\bigcap_{p=0}^lX_{\delta,\alpha,p,l-p},
$
and introduce the norm
\[\|f\|_{\delta,\alpha,l}:=\Biggl(\sum_{p=0}^l
  \|f\|_{\delta,\alpha,p,l-p}^2\Biggr)^{1/2}, \qquad 
   f\in X_{\delta,\alpha,l}.
\]
As we will see in the course of the following, however, it is not
possible to control solutions of our equations in the spaces
$X_{\delta,\alpha,l}$ directly. It turns out that we must use spaces
$(\tilde X_{\delta,\alpha,l}, \|\cdot\|_{\delta,\alpha,l}^\sim)$
instead. These are defined as earlier, but in the norm
$\|f\|_{\delta,\alpha,l}^\sim$ of some function $f$, the highest
spatial derivative term $\partial^l_x f$ is weighted with the
additional factor $t^{\beta+1}$. Here $\beta$ is the exponent of the
characteristic speed given by \Eqref{eq:behaviorofk}.  It is easy to
see under the earlier conditions that also $(\tilde
X_{\delta,\alpha,l}, \|\cdot\|_{\delta,\alpha,l}^\sim)$ are Banach
spaces.  We also note that $X_{\delta,\alpha,l}\subset\tilde
X_{\delta,\alpha,l}$. Let us also define
$X_{\delta,\alpha,\infty}:=\bigcap_{l=0}^\infty X_{\delta,\alpha,l}, $
and note that $X_{\delta,\alpha,\infty}=\bigcap_{l=0}^\infty \tilde
X_{\delta,\alpha,l}$.

For $w\in \tilde X_{\delta,\alpha,l}$ we set
\[ 
  w_\eta(t,x)=t^{-\lambda_2(x)}\int_{-\infty}^\infty\int_0^\infty 
  t^{\lambda_2(y)}w(s,y)k_{\eta}\Bigl(\log\frac st\Bigr) k_\eta(x-y) 
  \frac 1s dsdy.
\]
Here, $k_\eta:\R\rightarrow\R_+$ is a smooth kernel supported in
$[-\eta,\eta]$, satisfying $\int_\R k_\eta(x)dx=1$ for all positive
$\eta$. Then $w_\eta$ is an element of $\tilde
X_{\delta,\alpha-\epsilon,l}\cap C^\infty((0,\delta]\times U)$ for
every $\epsilon>0$. Furthermore, the sequence of such mollified
functions $w_\eta$ in the limit $\eta\rightarrow 0$ converges to $w$
in the norm $\|\cdot\|_{\delta,\alpha-\epsilon,l}^\sim$. Hence any
element in $\tilde X_{\delta,\alpha,l}$ can be approximated by smooth
functions.

We mentioned earlier that we are interested in solving the
``forward problem'', 
 referred to as the \textbf{singular
  initial value problem} (SIVP) -- in this paper. More precisely, we
need to guess a leading-order term $u$ of solutions $v$ to
\Eqref{eq:secondorderFuchsianHyp} so that the \textbf{remainder}
\begin{equation*} 
  w(t,x):=v(t,x)-u(t,x),
\end{equation*}
can be interpreted as ``higher order'' in $t$ at $t=0$. By this we
mean that $w$ is an element in $X_{\delta,\alpha,l}$ for some
(sufficiently large) $\alpha>0$ on a small time interval
$(0,\delta]$. If for a given $u$ such a solution $v$ exists then we
say that \textbf{$v$ obeys the leading-order behavior given by
  $u$}. Often $u$ will be parametrized by certain free functions which
we call \textbf{asymptotic data}, see below.  For later convenience,
we introduce the operator $F$ as
\begin{equation}
  \label{eq:definitionF}
  F[w](t,x):=f[u+w](t,x).
\end{equation}

\subsection{Canonical leading-order term}
\label{sec:canonicallot}

The main first step for solving the singular initial value problem is
to guess a leading-order term $u$. In some applications this can be
very tricky, but in many situations, which we will be most interested
in in this paper, one can make a canonical guess. These situations are
described heuristically as follows.

\subsubsection*{Canonical two-term expansion}

Consider the principal part operator $\tilde L$ in \Eqref{eq:defLPDE}
and note that it incorporates certain lower derivative terms. The
reason for writing $\tilde L$ like this is that we expect in many
cases that these terms are significant and of leading-order at the
singularity $t=0$. In contrast, the source-term and spatial
derivatives can often be anticipated as negligible in some sense under
suitable assumptions below. This is motivated by the BKL conjecture in
general relativity. In order to make this more
concrete, let us assume that the leading-order term is an exact
solution of the system of ordinary differential equations
(parametrized by $x$), which is obtained when all terms in the
equation, except for those given by $\tilde L$, are set to zero. We
refer to this leading-order term $u$ as the 
``canonical
  leading-order term'' or the \textbf{canonical two-term expansion}:
\begin{equation}
  \label{eq:splitv}
  u(t,x)=
  \begin{cases}
    u_*(x)\,t^{-a(x)}\log t+u_{**}(x)\,t^{-a(x)},
    & \quad (a(x))^2=b(x),\\
    u_*(x)\,t^{-\lambda_1(x)}+u_{**}(x)\,t^{-\lambda_2(x)},
    & \quad (a(x))^2 \neq b(x),
  \end{cases}    
\end{equation}
for some freely prescribed \textbf{asymptotic data} $u_*$ and $u_{**}
\in H^{m'}(U)$, where $m'$ is some non-negative integer. We refer to
this as the ``Fuchsian heuristics'' because the leading-order
behavior will be determined by Fuchsian ordinary differential
equations.

We clearly see the dependence of the expected leading-order behavior
at $t=0$ on the coefficients of the principal part of the equation. If
the roots $\lambda_1$ and $\lambda_2$ are real and distinct, i.e.\ if
$a^2>b$, we expect a {\sl power-law} behavior. In the degenerate case
$\lambda_1=\lambda_2$, i.e.\ if $a^2=b$, we expect a {\sl logarithmic}
behavior. Finally, when $\lambda_1$ and $\lambda_2$ are complex for
$a^2<b$, the solution is expected to have an {\sl oscillatory}
behavior at $t=0$ of the form
\[u(t,x)=t^{-a(x)}\bigl(\tilde u_*\cos(\lambda_I(x)\log
t)+\tilde u_{**}\sin(\lambda_I(x)\log t)\bigr)+\ldots
\] 
for some real coefficient functions $\tilde u_*(x)$ and $\tilde u_{**}(x)$;
note that in this case, $\lambda_1=\bar\lambda_2=a+i\lambda_I$ with
$\lambda_I:=\sqrt{b^2-a}$.

If the coefficients of the equations are such that there is a
continuous transition between the two cases in \Eqref{eq:splitv}, then
the asymptotic data functions $u_*$ and $u_{**}$ must be renormalized
as follows. Define $\Gamma(x):=\sqrt{a(x)^2-b(x)}$ which might be real
or imaginary dependent on the values of the coefficients. If there are
points $x_0\in U$ so that $\Gamma(x_0)=0$ and other points $x_1\in U$
with $\Gamma(x_1)\not=0$, then let us set
\begin{equation} 
  \label{eq:renormalizeddata}
  u_*(x)=\frac{\hat u_*(x)-\hat u_{**}(x)/\Gamma(x)}2,
  \qquad
  \quad
  u_{**}(x)=\frac{\hat u_*(x)+\hat u_{**}(x)/\Gamma(x)}2,
\end{equation} 
and choose $\hat u_*(x)$, $\hat u_{**}(x)$ as asymptotic data
functions. This guarantees that $u(t,x)$ given by \Eqref{eq:splitv} is
smooth for all $t>0$.

\subsubsection*{Higher-order canonical expansions}

For some applications we require expansions of the solutions at $t=0$
with more than two terms in order to describe the leading-order
behavior. A particular important example is the Gowdy case
in \Sectionref{sec:Gowdy}. Following Rendall \cite{Rendall00}, those
can be constructed as follows, without going into the
details. Consider the Fuchsian ODE case of
\Eqref{eq:secondorderFuchsianHyp}, written here for the scalar case
only,
\[D^2 v(t,x)+2a(x) Dv(t,x)+b(x)v(t,x)=f[v](t,x),
\] 
where $x$ is interpreted as a parameter. Let first
$f[v](t,x)=f_0(t,x)$ be a given function. Under suitable decay
assumption\footnote{Details can be found in
  \cite{BeyerLeFloch1,BeyerLeFloch2}.} on $f_0$ at $t=0$, there exists
a unique solution $v$ of this equation obeying the canonical two-term
expansion $u$ given by \Eqref{eq:splitv} for given asymptotic data
functions $u_*$ and $u_{**}$. Let $H$ be the operator mapping $f_0$ to
the remainder $w=v-u$ of the solution $v$. Now consider an arbitrary
source-term $f[v]$ and let the operator $F$ be as defined in
\Eqref{eq:definitionF} for the given function $u$. Let $w_1\equiv 0$
and
\[w_{j+1}:=H\circ F[w_j],\quad j\in\N.
\] 
Finally, set $v_j=u+w_j$ for all $j\in\N$. Clearly, $v_1=u$. One finds
that the order of $v_{j+1}-v_j$ in $t$ at $t=0$ increases with $j$.
Hence $v_j$ can be interpreted as an expansion of the solution at
$t=0$ whose order in $t$ increases with $j$.  Moreover, it turns out
that the order of the residual in $t$ at $t=0$, obtained when $v_j$ is
plugged into the equation, increases with $j$. Thus $v_j$ can be
considered as an asymptotic solution of the Fuchsian ODE. In many
situations it is thus meaningful to use $v_j$ as the canonical
leading-order term $u$ for any given $j\in\N$.

\subsubsection*{Limitations of this heuristics}

At this stage it is of course highly unclear under which conditions
there exist solutions of the equations which obey the leading-order
term $u$ given by \Eqref{eq:splitv} or any of its higher-order version
$v_j$ constructed before. The resolution of this problem will be
central to this paper. In many applications in general relativity, the
canonical two-term expansion is the correct guess for the
leading-order term, if the asymptotic data are consistent with
constraint equations implied by Einstein's field equations. However,
we know of several cases when the operator $\tilde L$ does not give
rise to the dominant term at $t=0$. For example for the Gowdy case,
nonlinear terms from the source-term need to be taken into
account. As we discuss there, however, the problem can be reduced to
the canonical case by adding a certain term to the equation. For Gowdy
solutions with spikes
\cite{BergerMoncrief,RendallWeaver,Lim,Andersson2,Andersson3,Lim2,Uggla,Ringstrom4,Ringstrom7}
the situation is significantly more complicated because then, other
nonlinear terms and spatial derivative terms can become
significant. Another important example is given by the mixmaster
dynamics
\cite{Damour,Andersson,Andersson3,Heinzle,Uggla,Ringstrom7}. There,
one has to control a complicated interplay between nonlinear terms in
the source-term in order to fix the leading-order term.  In general
when the equations are generalized to time-dependent coefficients in
the principal part or to the quasi-linear case, the notion of
canonical two-term expansions will apply only under suitable
conditions.


\section{Well-posedness theory and the Fuchsian numerical algorithm}  
\label{sec:existence}

\subsection{An approximation scheme}  

We begin with some notation. 
For $w\in \tilde X_{\delta,\alpha,1}$, the operator $L$ in \Eqref{eq:defLPDE}
is defined in the sense of distributions, only, via:  
$$
\aligned
   & \langle  \mathcal L[w],\phi\rangle
    \\
    &:=\int_{0}^\delta\int_\R
    t^{\Re\lambda_2(x)-\alpha} \Big(
    && -Dw(t,x) D\phi(t,x)
    +(2A(x)-\Re\lambda_2(x)+\alpha-1)\,Dw(t,x)\phi(t,x)\\
  &  &&+B(x)\,w(t,x)\phi(t,x)  
    + t K(t,x)\partial_x w(t,x) tK(t,x)\partial_x\phi(t,x)\\
  &  &&+(2t\partial_xK(t,x)+\partial_x\Re\lambda_2(x) K(t,x)t\log t)
    t K(t,x)\partial_xw(t,x)\,\phi(t,x)\Big)\,\,dx dt,
\endaligned
$$
where $\phi$ is any test function, i.e.~a real-valued
$C^\infty$-function on $(0,\delta]\times\R$ together with some $T\in
(0,\delta)$ and a compact (i.e.~closed and bounded) set $K\in\R$ so that $\phi(t,x)=0$ for all
$t>T$ and $x\not\in K$, and each derivative of $\phi$ has a finite
(not necessarily vanishing) limit at $t=0$ for every $x\in U$. 
For our later discussion, we note that for any
given test function $\phi$, the linear functional $\langle\mathcal
L[\cdot],\phi\rangle: \tilde X_{\delta,\alpha,1}\rightarrow\R$ is continuous
with respect to the norm $\|\cdot\|_{\delta,\alpha,1}^\sim$. This is the
main reason to include the factor $t^{\Re\lambda_2(x)-\alpha}$ in the
definition of $\mathcal L$.

If the operator $F$ defined by \Eqref{eq:definitionF} for a given
leading-order term $u$ gives rise to a map $\tilde
X_{\delta,\alpha,1}\rightarrow X_{\delta,\alpha,0}$, where $w\mapsto
F[w]$, it is meaningful to define its weak form by (for all test
functions $\phi$)
\[\langle\mathcal F[w],\phi\rangle:=
\int_{0}^\delta\int_\R t^{\Re\lambda_2(x)-\alpha}
F[w](t,x)\phi(t,x)dxdt. 
\]

\begin{definition}[Weak solutions of second-order hyperbolic Fuchsian
  systems] 
  Let $u$ be a given function and $\delta, \alpha>0$ be constants. 
  Then, one says 
  that $w\in \tilde X_{\delta,\alpha,1}$ is a weak solution to the second-order
  hyperbolic Fuchsian equation \Eqref{eq:secondorderFuchsianHyp},
  provided
  \begin{equation*}     
    \mathcal P[w]:=
    \mathcal L[w]+\mathcal L[u]-\mathcal F[w]
    =0.
  \end{equation*}
\end{definition} 


Let us now start our discussion with the linear case of second-order
hyperbolic Fuchsian equations and introduce our new approximation
scheme.  The following conditions are assumed: 

\begin{enumerate}
\item Vanishing leading-order part: $u\equiv 0$.
\item Linear source-term:
  \begin{equation}
    \label{eq:linearinhom2}
    F[w](t,x)=f_0(t,x)+f_1(t,x) w+f_2(t,x) Dw+f_3(t,x) t k \partial_x
    w,
  \end{equation}  
  with given functions $f_0$, $f_1$, $f_2$, $f_3$, so that
  $f_1$,
  $f_2$, $f_3$ are smooth spatially periodic on $(0,\delta]\times U$, and
  near $t=0$
  \begin{equation}
    \label{eq:linearinhomDecay}
    \sup_{x\in\bar U}f_a(t,x)=O(t^\mu), \qquad a=1,2,3,
  \end{equation}
  for some constant $\mu>0$. 
\end{enumerate} 
We have not made any assumptions for the function $f_0$ yet, since in
the following discussion this function will play a different role than
$f_1$, $f_2$, $f_3$.  Moreover, no loss of generality is implied by
the condition $u\equiv 0$, since the general case can be recovered by
absorbing $L[u]$ into the function $f_0$.

Under these assumptions, we pose the question whether there exists a
unique weak solution $w$ in $\tilde X_{\delta,\alpha,1}$ for some
$\delta,\alpha>0$ of the given second-order hyperbolic equation.  The
main idea here is our new approximation scheme. We approximate a
solution of the {\sl singular} value problem by a sequence of
solutions of the {\sl regular} initial value problem.

\begin{definition}[Regular initial value problem (RIVP)]
  Fix $t_0\in (0,\delta]$ and some smooth periodic functions
  $g,h:U\rightarrow\R$, and suppose that the right-hand side is of the
  form \Eqref{eq:linearinhom2} with given smooth spatially periodic
  functions $f_0$, $f_1$, $f_2$, $f_3$ on $[t_0,\delta]\times U$.
  Then, $w:[t_0,\delta]\times U\rightarrow\R$ is called a solution of
  the {\bf regular initial value problem} associated with the
  ``regular data'' $g, h$, if \Eqref{eq:secondorderFuchsianHyp}
  holds everywhere on $(t_0,\delta]\times U$ and, moreover, the
  remainder $w:=v-u$, for some function $u$, satisfies
  \[w(t_0,x)=g(x), \quad \partial_t w(t_0,x)=h(x).\]
\end{definition}
For the regular initial value problem, we indeed assume that $f_0$ is
smooth, just as $f_1$, $f_2$ and $f_3$. By the general theory of
linear hyperbolic equations, the regular initial value problem is
well-posed, in the sense that there exists a unique smooth solution
$w$ defined on $[t_0,\delta]$ for any choice of smooth regular
data. Let $u$ be a choice of leading-order term. Let $(t_n)$ be a
sequence of positive times converging to zero so that each $t_n$ is
smaller than some $\delta>0$. For each $n\in N$, we construct an
approximate solutions $w_n$ of the singular initial value problem as
follows. Let $w_n\equiv 0$ on the time interval $(0,t_n]$. On the time
interval $[t_n,\delta]$ we set $w_n$ to be the solution of the regular
initial value problem with initial time $t_n$ and zero regular
data. Hence, $w_n$ is in $C^1( (0,\delta]\times U)\cap \tilde
X_{\delta,\alpha,1}$.

The central result of this section is that this sequence of
approximate solutions converges to the solution of the singular
initial value problem under suitable conditions.

\begin{proposition}[Existence of solutions of the linear singular
  initial value problem in $\tilde X_{\delta,\alpha,1}$]
  \label{prop:existenceSVIPlinear}
  Under the assumptions \Eqref{eq:linearinhom2} and
  \eqref{eq:linearinhomDecay}, the sequence $(w_n)$ of approximate
  solutions with initial times $(t_n)\rightarrow 0$ converges to the
  unique solution $w\in \tilde X_{\delta,\alpha,1}$ of the singular
  initial value problem for given $\delta,\alpha>0$ provided:
  \begin{enumerate}
  \item The matrix
    \begin{equation}
      \label{eq:defNfinal3}
      N:=
      \begin{pmatrix}
        \Re(\lambda_1-\lambda_2)+\alpha & ((\Im\lambda_1)^2/\eta-\eta)/2 & 0 \\
        ((\Im\lambda_1)^2/\eta-\eta)/2 
        & \alpha 
        & t\partial_x{k}-\partial_x\Re(\lambda_1-\lambda_2)(t k\log t) \\
        0 & t\partial_x{k}-\partial_x\Re(\lambda_1-\lambda_2)(t k\log t) 
        & \Re(\lambda_1-\lambda_2)+\alpha-1-Dk/k
      \end{pmatrix}
    \end{equation}
    is positive semidefinite at each
    $(t,x)\in (0,\delta)\times U$ for a constant $\eta>0$.
  \item The source-term function $f_0$ is in
    $X_{\delta,\alpha+\epsilon,0}$ for some $\epsilon>0$.
  \end{enumerate}
  Then, the solution operator
  $\mathbb H: X_{\delta,\alpha+\epsilon,0}
  \rightarrow \tilde X_{\delta,\alpha,1},\quad f_0\mapsto w,
  $ 
  is
  continuous and there exists a finite constant $C_\epsilon>0$ so that
  \begin{equation*}
    \|\mathbb H[f_0]\|_{\delta,\alpha,1}^\sim \le \delta^{\epsilon}
    C_\epsilon \|f_0\|_{\delta,\alpha+\epsilon,0},
  \end{equation*}
  for all such $f_0$.  The constant $C_\epsilon$ can depend on $\delta$,
  but is bounded for all small $\delta$. The approximate solutions $w_n$
  satisfy the following error estimate for all $n,m\in\N$:
  \begin{equation*}
    \|w_m-w_n\|_{\delta,\alpha,1}^\sim\le C|G(t_n)-G(t_m)|,
  \end{equation*}
  where $C>0$ is a constant and
  \begin{equation*}
    G(t):=\int_{0}^{t} s^{-1}\|s^{\Re\lambda_2-\alpha}f_0(s,\cdot)\|_{L^2(U)}ds.
  \end{equation*}
\end{proposition}
We call $N$ the \textbf{energy dissipation matrix}. We have assumed
that $\alpha$ is a positive constant. If, however, $\alpha$ is a
positive spatially periodic function in $C^1(U)$, the definition of
the spaces $X_{\delta,\alpha,k}$ and $\tilde X_{\delta,\alpha,k}$
remains the same, and only the $(2,3)$- and $(3,2)$-components of the
energy dissipation matrix $N$ change to
$t\partial_x{k}-\partial_x(\Re(\lambda_1-\lambda_2)+\alpha)(t k\log
t)$. In the following, we continue to assume that $\alpha$ is a
constant in order to keep the presentation as simple as possible, but
we stress that all following results hold (with this slight change of
$N$) if $\alpha$ is a function, and hence no new difficulties arise.

The proof is based on controlling the energy of the approximate
solutions.  Let us restrict our presentation here to the scalar case
$n=1$ for this whole section; the general case can be obtained with
the same ideas.  Choose $\delta,\alpha>0$ and let $w\in
C^{1}((0,\delta]\times U)$ be a spatially periodic
function. Then, we define its {\bf energy} at the time $t\in
(0,\delta]$ by \be
\label{eq:defenergy}
\aligned
E[w](t):=&  e^{-\kappa t^\gamma} 
\int_U t^{2(\lambda_2(x)-\alpha)} \, e[w](t,x) \, dx,
\\
e[w](t,x) :=&  
\frac 12\left((\eta\, w(t,x))^2+(Dw(t,x))^2 
  +(t k(t,x)\partial_x w(t,x))^2\right),
\endaligned
\ee 
for some constants $\kappa\ge 0$, $\gamma>0$ and $\eta>0$.  
For convenience, we also introduce the following notation. For any
scalar-valued function $w$, we define the vector-valued function
\be
  \label{eq:hatw}
  \widehat w(t,x):=t^{\Re\lambda_2(x)-\alpha}
  (\eta w(t,x),Dw(t,x),t k(t,x)\partial_x w(t,x)),
\ee
involving the same constants as in the energy.
Then, we can write
\begin{equation*}
  E[w](t)=\frac 12 e^{-\kappa t^\gamma}\|\widehat
  w(t,\cdot)\|^2_{L^2(U)},
\end{equation*}
the norm here being the Euclidean $L^2$-norm for vector-valued
functions in $x$.  It is important to realize that, provided $\eta>0$,
the expression $\sup_{0<t\le\delta} \|\widehat w(t,\cdot)\|_{L^2(U)}$
for functions of the form \Eqref{eq:hatw} yields a norm which is
equivalent to $\|\cdot\|_{\delta,\alpha,1}^\sim$, thanks to
\Eqref{eq:behaviorofk}. Therefore, the energy \Eqref{eq:defenergy} is
of relevance for the space $\tilde X_{\delta,\alpha,1}$.

Of central importance for the proof of
\Propref{prop:existenceSVIPlinear} are energy estimates for the
regular initial value problem.
\begin{lemma}[Fundamental energy estimate for the regular initial
  value problem]
  Suppose that the source term is of the form \Eqref{eq:linearinhom2}
  with the condition \Eqref{eq:linearinhomDecay} and that the energy
  dissipation matrix \Eqref{eq:defNfinal3} is positive semidefinite on
  $(0,\delta]\times U$ for given constants $\alpha,\eta>0$ and a
  sufficiently small $\delta>0$. Then there exist constants
  $C,\kappa,\gamma>0$, independent of the choice of $t_0\in
  (0,\delta]$, so that for all solutions
  $w$ of the regular initial value problem with smooth regular data at
  $t=t_0$, we have 
  \begin{equation*}
   \aligned
    &      \|\widehat{w}(t,\cdot)\|_{L^2(U)}
    \leq  
    C \, e^{\frac 12\kappa(t^\gamma-t_0^\gamma)}
    \Bigl(\|\widehat{w}(t_0,\cdot)\|_{L^2(U)}
    +\int_{t_0}^t
    s^{-1}\|s^{\Re\lambda_2-\alpha}
    f_0(s,\cdot)\|_{L^2(U)}ds\Bigr),
   \endaligned
 \end{equation*}
 for all $t\in [t_0,\delta]$.
\end{lemma}

The proof of this lemma is standard. However, one has to confirm that
the energy stays uniformly finite at $t=0$, and this is guaranteed by
the positivity condition for $N$ in \Eqref{eq:defNfinal3}.  Moreover,
this results demonstrates the importance of the assumption
$\beta(x)>-1$ in \Eqref{eq:behaviorofk}. Namely, if $\beta(x)\le -1$
at a point $x\in U$, then for any choice of $\alpha$ and $\eta$, the
matrix $N$ would not be positive semidefinite for small $t$ at
$x$. While the energy estimate would still be true for a given $t_0$,
we would nevertheless lose uniformity of the constants in the
estimates with respect to $t_0$. We stress that this uniformity is
crucial in the proof of \Propref{prop:existenceSVIPlinear}. All proofs
and more details can be found in \cite{BeyerLeFloch1}.

\subsection{Nonlinear theory}
In turns out that the space $\tilde X_{\delta, \alpha,1}$ is too large
for our nonlinear theory. Namely, we will need to require a Lipschitz
property of the source-term for which this space rules out natural
nonlinearities, for instance quadratic ones.  Moreover, the statement
that the solution of the Fuchsian equation $w$ is an element of
$\tilde X_{\delta, \alpha,1}$ yields only weak information about the
behavior of the first spatial derivative at $t=0$, which is indeed not
sufficient to interpret $w$ and \textit{all} its first derivatives as
the \textit{remainder} of the solution. We resolve these problems by
going to $\tilde X_{\delta, \alpha,k}$ for some $k>1$. The first issue
above disappears for $k=2$ in one spatial dimension.  In some
applications, the spaces $\tilde X_{\delta, \alpha,k}$ still impose a
too strong restriction due to the weak control of the highest spatial
derivative. In such a case we are required to formulate the theory in
the space $X_{\delta,\alpha,\infty}$.

The first step is to reconsider the linear case and derive a result
analogous to \Propref{prop:existenceSVIPlinear} in the spaces $\tilde
X_{\delta, \alpha,k}$ for arbitrary $k$ by making suitable stronger
assumptions on $f_0$. It turns out that higher spatial derivatives of
the solutions may involve additional logarithmic terms in $t$. As a
consequence we find that the energy dissipation matrix must be assumed
to be positive definite instead of positive semidefinite. With this at
hand, the central result of this section, which we write for $k=2$ for
definiteness, is the following.

\begin{proposition}[Existence of solutions of the nonlinear singular
  initial value problem in $\tilde X_{\delta,\alpha,2}$]
  \label{prop:existenceSVIPnonlinear2}
  Suppose that we can choose $\alpha>0$ so that the energy dissipation
  matrix \Eqref{eq:defNfinal3} is positive definite at each
  $(t,x)\in (0,\delta)\times U$ for a constant $\eta>0$.  Suppose that
  $u\equiv 0$ and that the operator $F$ has the following Lipschitz
  continuity property: For a constant $\epsilon>0$ and all
  sufficiently small $\delta$, the operator $F$ maps $\tilde
  X_{\delta,\alpha,2}$ into $X_{\delta,\alpha+\epsilon,1}$ and,
  moreover, for each $r>0$ there exists $\widehat C>0$ (independent of
  $\delta$) so that 
  \be
  \label{eq:lipschitzFPDE2}
  \|F[w]-F[\widetilde w]\|_{\delta,\alpha+\epsilon,1} \le \widehat C \, 
  \|w-\widetilde w\|_{\delta,\alpha,2}^\sim
  \ee
  for all $w,\widetilde w\in \overline{B_r(0)}\subset
  \tilde X_{\delta,\alpha,2}$.
  Then,    
  there exists a unique
  solution $w\in \tilde X_{\delta,\alpha,2}$ of the singular initial value
  problem.
\end{proposition}

\begin{proof}[Proof of \Propref{prop:existenceSVIPnonlinear2}:]  
  We define the operator $\mathbb G:=\mathbb H\circ F$. Here $\mathbb
  H$ is the solution operator as in \Propref{prop:existenceSVIPlinear}
  and $F$ is the source-term operator as before. We find that under
  the hypothesis, the operator $\mathbb G$ is a contraction on closed
  bounded subsets of $\tilde X_{\delta,\alpha,2}$ if $\delta$ is a
  sufficiently small. Hence the iteration sequence defined by
  $w_{j+1}=\mathbb G[w_j]$ for $j\ge 1$ and, say, $w_1=0$ converges to
  a fixed point $w\in \tilde X_{\delta,\alpha,2}$ with respect to the
  norm $\|\cdot\|_{\delta,\alpha,2}^\sim$. Because of the properties
  of $\mathbb H$, a fixed point of $\mathbb G$ is a solution of the
  SIVP. Hence, we have shown existence of solutions. Uniqueness can be
  shown as follows. Given any other solution $\tilde w$ in $\tilde
  X_{\delta,\alpha,2}$, it is a fixed point of the iteration
  $w_{j+1}=\mathbb G[w_j]$. Because $\mathbb G$ is a contraction,
  there, however, only exists one fixed point, and hence $\tilde w=w$.
\end{proof}

For the analogous result of infinite differentiability, we only need
to substitute the Lipschitz continuity property in the previous result
as follows. For a constant $\epsilon>0$, every sufficiently small
$\delta>0$ and every non-negative integer $k$, the operator $F$ maps
$X_{\delta,\alpha,k+1}$ into $X_{\delta,\alpha+\epsilon,k}$ and,
moreover, for each $r>0$, there exists $\widehat C>0$ (independent of
$\delta$) so that
\begin{equation}
  \|F[w]-F[\widetilde w]\|_{\delta,\alpha+\epsilon,k} \le \widehat C \, 
  \|w-\widetilde w\|_{\delta,\alpha,k+1}^\sim
\end{equation}
for all $w,\widetilde w\in \overline{B_r(0)}\cap
X_{\delta,\alpha,k+1}\subset \tilde X_{\delta,\alpha,k+1}$.  Then,
there exists a unique solution $w\in X_{\delta,\alpha,\infty}$ of the
singular initial value problem. In order to prove this result, we
first proceed as in the previous proposition for finitely many
derivatives. The only remaining task is to show that for
$k\rightarrow\infty$, we are allowed to choose some non-vanishing
$\delta$. This can be done with a standard argument for hyperbolic
equations. The set $\overline{B_r(0)}$ is defined with respect to the
norm $\|\cdot\|^\sim_{\delta,\alpha,k+1}$.  We note that the constant
$\widehat C$ is allowed to depend on $k$ and is not required to be
bounded for $k\rightarrow\infty$. Note that the Lipschitz estimate
involves the norm $\|\cdot\|_{\delta,\alpha,k+1}^\sim$, while the
elements for which this estimates needs to be satisfied are required
to be only in the subspace $X_{\delta,\alpha,k+1}$ of $\tilde
X_{\delta,\alpha,k+1}$. The main advantage of the $C^\infty$ result
over the finite differentiability case is that we only need to check
that $F$ maps $X_{\delta,\alpha,k+1}$ into
$X_{\delta,\alpha+\epsilon,k}$, instead of $\tilde
X_{\delta,\alpha,k+1}$ into $X_{\delta,\alpha+\epsilon,k}$.

\subsection{Standard singular initial value problem}
The following discussion is devoted to the case when the function $u$
is given by the canonical two-term expansion \Eqref{eq:splitv}.  In
this case, we speak of the \textbf{standard singular initial value
  problem}.

\begin{theorem}[Well-posedness of the standard singular initial value problem
  in $\tilde X_{\delta,\alpha,2}$]
  \label{th:well-posednessSIVP}
  Given arbitrary asymptotic data $u_*,u_{**}\in H^3(U)$, the
  standard singular initial value problem admits a unique solution
  $w\in \tilde X_{\delta,\alpha,2}$ for $\alpha,\delta>0$, provided
  $\delta$ is sufficiently small and the following conditions hold:
  \begin{enumerate}
  \item{Positivity condition.} Suppose that we can choose $\alpha>0$
    so that the energy dissipation matrix \Eqref{eq:defNfinal3} is
    positive definite at each $(t,x)\in (0,\delta)\times U$ for a
    constant $\eta>0$.
  \item{Lipschitz continuity property.} For the given $\alpha>0$, the
    operator $F$ satisfies the Lipschitz continuity property stated in
    \Propref{prop:existenceSVIPnonlinear2} for all asymptotic data
    $u_*,u_{**}\in H^3(U)$ for some $\epsilon>0$.
  \item{Integrability condition.} The constants $\alpha$ and
    $\epsilon$ satisfy
    \begin{equation}
      \label{eq:intbeta}
      \alpha+\epsilon<2(\beta(x)+1)-\Re(\lambda_1(x)-\lambda_2(x)), 
      \qquad  x \in U.
    \end{equation}  
  \end{enumerate}
\end{theorem}

We note that the regularity assumptions on the asymptotic data can
certainly be improved. An analogous theorem can be formulated for the
$C^\infty$-case.

\begin{proof}
  We can apply \Propref{prop:existenceSVIPnonlinear2} if we are able
  to control the additional contribution of the term $L[u]$ which has
  to be considered as part of the source-term. It has no contribution
  to the Lipschitz estimate \Eqref{eq:lipschitzFPDE2}, but we have to
  guarantee that under these hypotheses, $L[u]\in
  X_{\delta,\alpha+\epsilon,1}$ for the given constant
  $\epsilon$. This is indeed the case if \Eqref{eq:intbeta} holds.
\end{proof}

\begin{example}
  \label{ex:EPD}
  Consider the second-order hyperbolic Fuchsian equation
  \[D^2 v-\lambda Dv-t^2\partial_x^2 v=0,
  \] 
  with a constant $\lambda$. This is the Euler-Poisson-Darboux
  equation. In the standard notation it is
  \[
  \partial_t^2 v-\partial_x^2 v=\frac 1t(\lambda-1)\partial_t v.
  \]
  Note that $\lambda=1$ is the standard wave equation, and in
  this case, the standard singular initial value problem reduces to
  the standard Cauchy problem.
  \begin{enumerate}
  \item Case $\lambda\ge 0$. With our notation, we have $\lambda_1=0$,
    $\lambda_2=-\lambda$, $\beta\equiv 0$, $\nu\equiv 1$ and $f\equiv
    0$. Thus the leading-order term for the standard singular initial value
    problem is
    \[u(t,x)=
    \begin{cases}
      u_*(x)+u_{**}(x) t^{\lambda} & \lambda>0,\\
      u_*(x)\log t+u_{**}(x) & \lambda=0.
    \end{cases}
    \]
    The positivity condition of the energy dissipation matrix
    \Eqref{eq:defNfinal3} is satisfied precisely for $\alpha\ge
    1-\lambda$ and all sufficiently small $\eta>0$. The integrability
    condition \Eqref{eq:intbeta} is satisfied precisely for $\lambda<
    2-\alpha$. Hence, our previous proposition implies that the
    singular initial value problem is well-posed, provided \be
      \label{eq:wellposednessEPD}
      0\le\lambda<2.
    \ee
    Namely, in this case there
    exists a solution $w$ in $\tilde X_{\delta,\alpha,2}$ for some $\alpha>0$ for
    arbitrary asymptotic data in $H^3(U)$.    
  \item Case $\lambda<0$. With our notation, we have
    $\lambda_1=|\lambda|$, $\lambda_2=0$, $\beta\equiv 0$,
    $\nu\equiv 1$ and $f\equiv 0$. The positivity condition of the
    energy dissipation matrix \Eqref{eq:defNfinal3} is satisfied
    precisely for $\alpha\ge 1-|\lambda|$ and all sufficiently
    small $\eta>0$. The integrability condition \Eqref{eq:intbeta} is
    satisfied precisely for $|\lambda|< 2-\alpha$. Hence, our
    previous proposition implies that the singular initial value
    problem is well-posed, provided
    \begin{equation*} 
      -2<\lambda<0.
    \end{equation*}
    Namely, in this case there exists a solution $w$ in
    $\tilde X_{\delta,\alpha,2}$ for some $\alpha>0$ for arbitrary asymptotic
    data in $H^3(U)$.
  \end{enumerate}
\end{example}
It turns out that general smooth solutions to the
Euler-Poisson-Darboux equation can be expressed explicitly by a
Fourier ansatz in $x$ and by Bessel functions in $t$. It is then easy
to check that \Eqref{eq:wellposednessEPD} (and similarly for
$\lambda<0$) is sharp: While for $0\le\lambda<2$, all solutions of the
equation behave consistently with the two-term expansion at $t=0$,
this is not the case for $\lambda\ge 2$ for general asymptotic
data. Hence the standard singular initial value problem is not
well-posed for $\lambda\ge 2$. This is completely consistent with our
heuristic discussion in \Sectionref{sec:canonicallot} underlying the
canonical guess for the leading-order term. If e.g.\ $\lambda=2$,
the assumption that the source-term $t^2\partial_x^2 v$ is negligible
at $t=0$ fails since it is of the same order in $t$ at $t=0$ as the
second leading-order term.  However, we can see in the proof of
\Theoremref{th:well-posednessSIVP} that in the special case $u_*=0$
(and arbitrary $u_{**}$), the integrability condition
\Eqref{eq:intbeta} can be relaxed.  For this special choice of data,
solutions to the singular initial value problem exist even for
$\lambda\ge 2$.

It turns out that, often in applications, the three conditions in
\Theoremref{th:well-posednessSIVP} cannot be satisfied
simultaneously. While it is often possible to find constants $\alpha$
and $\epsilon$ in accordance with the second and third condition, it
can turn out that the corresponding choice of $\alpha$ is too small to
make the energy dissipation matrix positive definite. In order to
circumvent this problem, the trick is to choose canonical expansions
of higher order $v_j$, see \Sectionref{sec:canonicallot}, as the
leading-order term $u$ for sufficiently large $j$. We refer to the
singular initial value problem based on this choice of leading-order
term as \textbf{singular initial value problem of order $j$}.
For $j=1$, it reduces to the standard singular initial value problem;
hence we will focus on the case $j\ge 2$ in the following. Note that,
if $w$ is a solution of the singular initial value problem of order
$j$, it is also a solution of the standard singular initial value
problem. However, if there is only one solution $w$ of the singular
initial value problem of order $j$ for given asymptotic data, it does
not necessarily mean that $w$ is the only solution of the standard
initial value problem for the same asymptotic data.

For the statement of the following theorem, we need the following notation. 
For all $w\in X_{\delta,\alpha,k}$ (or $w\in \tilde
X_{\delta,\alpha,k}$ respectively), we introduce the functions
$E_{\delta,\alpha,k}[w]:(0,\delta]\rightarrow\R$ (or $\tilde
E_{\delta,\alpha,k}[w]:(0,\delta]\rightarrow\R$ respectively) which
are defined in the same way as the respective norms, but the supremum
in $t$ has not been evaluated yet. In particular, this means that
$E_{\delta,\alpha,k}[w]$ (or $\tilde E_{\delta,\alpha,k}[w]$) is a
bounded continuous function on $(0,\delta]$.

\begin{theorem}[Well-posedness of the singular initial value problem
  of order $j$ in $\tilde X_{\delta,\alpha,2}$]
  \label{th:WellPosednessHigherOrderSIVP}
  Given any integer $j\ge 2$ and any asymptotic data
  $u_*,u_{**}\in H^{m_1}(U)$ with $m_1=2j+1$, there exists a unique
  solution $w\in \tilde X_{\delta,\alpha,2}$ of the singular initial
  value problem of order $j$ for some $\alpha>0$, provided
  \begin{enumerate}
  \item $F$ maps $\tilde X_{\delta,\tilde\alpha,m_1}$ into
    $X_{\delta,\tilde\alpha+\epsilon,m_1-1}$ for all asymptotic data
    $u_*,u_{**}\in H^{m_1}(U)$ for some $\epsilon>0$ and
    $\tilde\alpha$ given by 
    $
      \alpha:=\tilde\alpha+(j-2)\kappa\epsilon.
    $
    for an arbitrary $\kappa<1$.
  \item The characteristic speed satisfies
    \[2(\beta(x)+1)>\kappa\epsilon \quad\text{for all }x\in U\]
    for the same constant $\kappa$ chosen earlier.
  \item $F$ satisfies the following Lipschitz condition: for each
    $r>0$ there exists a constant $C>0$ (independent of $\delta$) so
    that
    \[E_{\delta,\tilde\alpha+\epsilon,1}[F[w]-F[\widetilde w]](t)
    \le C \tilde E_{\delta,\tilde\alpha,2}[w-\widetilde w](t)\]
    for all $t\in (0,\delta]$ and
    for all $w,\widetilde w\in \overline{B_r(0)}\subset
    \tilde X_{\delta,\tilde\alpha,2}$.
  \item The energy dissipation matrix \Eqref{eq:defNfinal3} (evaluated
    with $\alpha$) is positive definite at each $(t,x)\in
    (0,\delta)\times U$ for a constant $\eta>0$.
  \end{enumerate}
\end{theorem}

The third condition above is meaningful since both sides of the
inequality are continuous and bounded functions on $(0,\delta]$. Note
that this theorem can be formulated without difficulty for the
$C^\infty$-case and indeed leads to a simpler statement.

In effect we have obtained a value of $\alpha$ which increases with
$j$ and henceforth improves the positivity of the energy dissipation
matrix. The main prize to pay here is that the asymptotic data must be
sufficiently regular and that we must live with a loss of regularity
in space.

\subsection{The Fuchsian numerical algorithm}

\label{sec:numerical_scheme}

We proceed with the numerical implementation of our approximation
scheme.  For linear source-terms we have shown that the solution of
the singular initial value problem can be approximated by solutions to
the regular initial value problem. We have established an explicit error
estimate for these approximate solutions. For the nonlinear case, an
additional fixed point argument was necessary for the proof, but the
Lipschitz continuity condition should allow us to extend the error
estimates to nonlinear source-terms.

The regular initial value problem for second-order hyperbolic
equations corresponds to the standard initial value problem of a
system of (nonlinear) wave equations with initial time $t_0>0$, and
there exists a huge amount of numerical techniques for computing
solutions \cite{Kreiss2,LeVeque} 
However, a second-order Fuchsian equation written out
with the standard time-derivative $\del_t$ (instead of $D$) clearly
involves factors $1/t$ or $1/t^2$. Although these are finite for the
regular initial value problem, they still can cause severe numerical
problems when the initial time $t_0$ approaches zero, due to the
finite representation of numbers in a computer. In order to solve this
problem, we introduce a new time coordinate $ \tau:=\log t, $ and
observe that $D=\del_\tau$. For instance, the Euler-Poisson-Darboux
equation becomes
\begin{equation*}
  \del_\tau^2v-\lambda\, \del_\tau v-e^{2\tau}\del_x^2 v=0,
\end{equation*}
where $v$ is the unknown and $\lambda$ is a constant.  We have
achieved that there is no singular term in this equation; the main
price to pay, however, is that the singularity $t=0$ has been
``shifted to'' $\tau=-\infty$. Another disadvantage is that the
characteristic speed of this equation (defined with respect to the
$\tau$-coordinate) is $e^\tau$ and hence increases exponentially with
time.  For any explicit discretization scheme, we can thus expect that
the CFL-condition\footnote{The Courant-Friedrichs-Lewy (CFL) condition
  for the discretization of hyperbolic equations with explicit schemes
  is discussed, for instance, in \cite{Kreiss2}.} is always violated from some
time on. We must either adapt the time step to the increasing
characteristic speeds in $\tau$, or, when we decide to work with a
fixed time resolution, accept the fact that the numerical solution
will eventually become instable. However, this is not expected to be a
severe problem since one can compute the numerical solution with
respect to the $\tau$-variable until some finite positive time when
the numerical solution is still stable and then, if necessary, switch
to a discretization scheme based on the original $t$-variable. For all
the numerical solutions presented in this paper, however, this was not
necessary.

We can simplify the following discussion slightly by writing (and
implementing numerically) the equation not for the function $v$ but
for the remainder $w=v-u$. Henceforth we assume that $u$ is the
canonical two-term expansion determined by given asymptotic data. We
have to solve the equation for $w$ on a time interval
$[\tau_0,\delta]$ for some $\tau_0\in\R$ successively going to
$-\infty$ with regular data
\[w(\tau_0,x)=0,\qquad \quad \del_\tau w(\tau_0,x)=0,\qquad x\in U.
\]

Inspired by Kreiss et al.\ in \cite{Kreiss} and by the general idea of
the ``method of lines'', see \cite{Kreiss2}, we proceed as follows to
discretize the equation. First we consider second-order Fuchsian
ordinary differential equations (written for a scalar equation now for
simplicity)
$$
\del_\tau^2 w+2 a\,\del_\tau w+b\, w=f(\tau),
$$ 
where $f$ is a given function and the coefficients $a$ and $b$ are
constants. We discretize the time variable $\tau$ so that
$\tau_n:=\tau_0+n \Delta\tau$, $w_n:=w(\tau_n)$ and $f_n:=f(\tau_n)$
for some time step $\Delta\tau>0$ and $n\in\N$. Then the equation is
discretized in second-order accuracy as
\begin{equation}
  \label{eq:numericalschemeODE}
  \frac{w_{n+1}-2w_n+w_{n-1}}{(\Delta\tau)^2}
  +2a \frac{w_{n+1}-w_{n-1}}{2\Delta\tau}+b w_n=f_n.
\end{equation}
Solving this for $w_{n+1}$ allows to compute the solution $w$ at the
time $\tau_{n+1}$ from the solution at the given and previous time
$\tau_n$ and $\tau_{n-1}$, respectively. At the initial two time steps
$\tau_0$ and $\tau_1$, we set, consistently with the initial data for
$w$ at $\tau_0$ above,
\begin{equation}
  \label{eq:initializescheme}
  w_0=0,\quad w_1=\frac 12 (\Delta\tau)^2 f(\tau_0).
\end{equation}
We will refer to this scheme as the \textit{Fuchsian ODE solver}.

The idea of the method of lines for Fuchsian partial differential
equations is to discretize also the spatial domain with the spatial
grid spacing $\Delta x$ and to use our Fuchsian ODE solver to
integrate one step forward in time at each spatial grid point. The
source-term function $f$, which might now depend on the unknown itself
and its first derivatives, is then computed from the data on the
current or the previous time levels. Here we understand that spatial
derivatives are part of the source-term and are discretized by means
of the standard second-order centered stencil using periodic boundary
conditions.  A problem is that $f$, besides spatial derivatives, can
also involve time derivatives of the unknown $w$ (in fact this can
also be the case for Fuchsian ordinary differential equations when the
source term depends on the time derivative of the unknown). In order
to compute those time derivatives in second-order accuracy without
changing the stencil of the Fuchsian ODE solver, we made the following
choice. In the code we store the numerical solution not only on two
time levels, as it is necessary up to now for the scheme given by
\Eqref{eq:numericalschemeODE} and \eqref{eq:initializescheme}, but on
a further third past time level. The time derivatives in the
source-term can then be computed in second-order accuracy from data
exclusively at the present and previous time steps as follows
\[\del_{\tau} w(\tau_n)=\frac{3 w_n-4 w_{n-1}+w_{n-2}}{2\Delta\tau}
+O((\Delta\tau)^2).
\]
For this, we need to initialize three time levels at $\tau=\tau_0$ and
hence we set $w_2=2 (\Delta\tau)^2 f(\tau_0)$, in addition to
\Eqref{eq:initializescheme}.

\subsection{An example: the Euler-Poisson-Darboux equation}
We present numerical test results for the Euler-Poisson-Darboux
equation now. Recall from Example~\ref{ex:EPD} that the singular
initial value problem with two-term asymptotic data for this equation
is well-posed in particular for $0\le \lambda<2$, and in general
becomes ill-posed for $\lambda\ge 2$. The singular initial value
problem for $\lambda>0$ considers solutions of the form
\[v(t,x)=u_*(x)+u_{**}(x)t^\lambda+w(t,x),
\] 
with remainder $w$. For the purposes of this test, we choose the
asymptotic data $u_*=\cos x$, $u_{**}=0$. Note that in this case, this
leading-order behavior is consistent even with the case
$\lambda=0$. But according to our previous discussion, it is not
consistent with $\lambda=2$, and we expect that this becomes visible
in the numerical solutions. For $u_{**}=0$ and $0<\lambda<2$, we can
show that the leading-order behavior of the remainder at $t=0$ is
\begin{equation}
  \label{eq:leadingorderexact}
  w(t,x)=u_*(x)\left(-\frac 1{2(2-\lambda)}t^2 +\frac
  1{8(2-\lambda)(4-\lambda)}t^4+\ldots\right).
\end{equation}

\begin{figure}[ht]
  \subfigure[$\lambda=0.01$.]{\includegraphics[width=0.49\linewidth]%
    {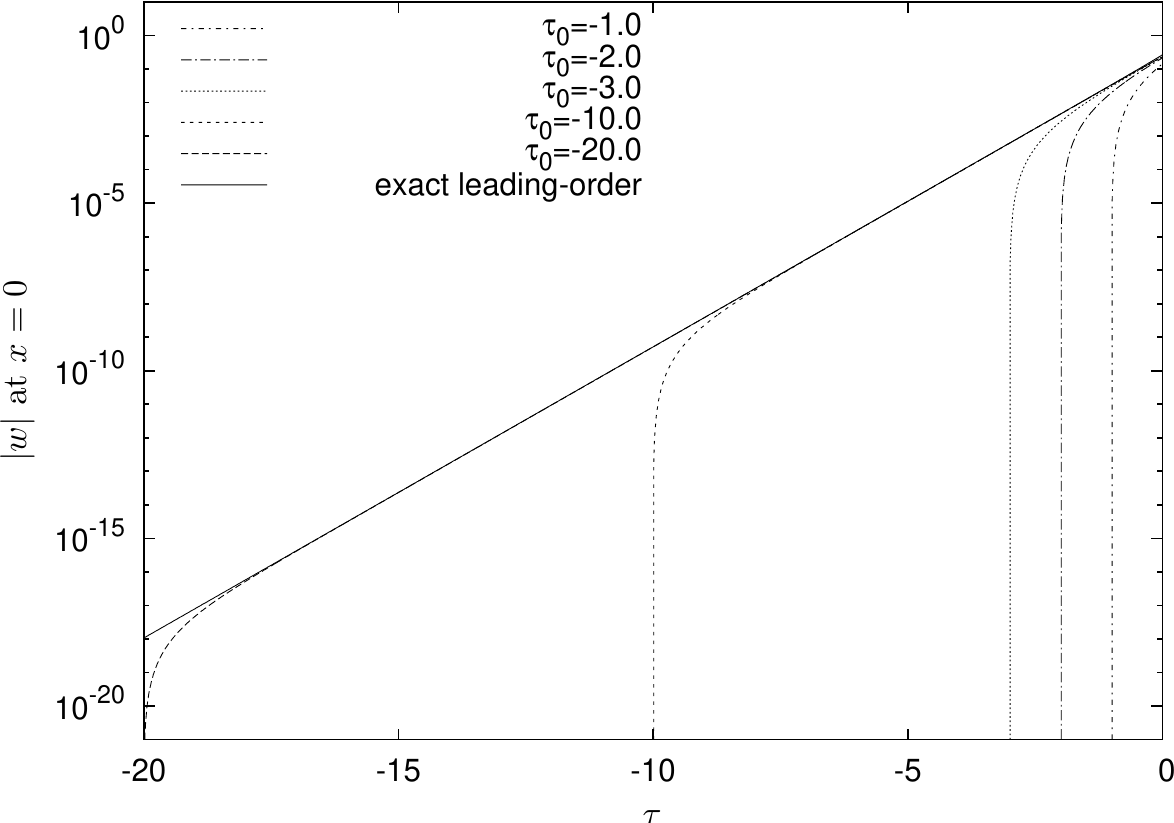}}
  \subfigure[$\lambda=1.00$.]{\includegraphics[width=0.49\linewidth]%
    {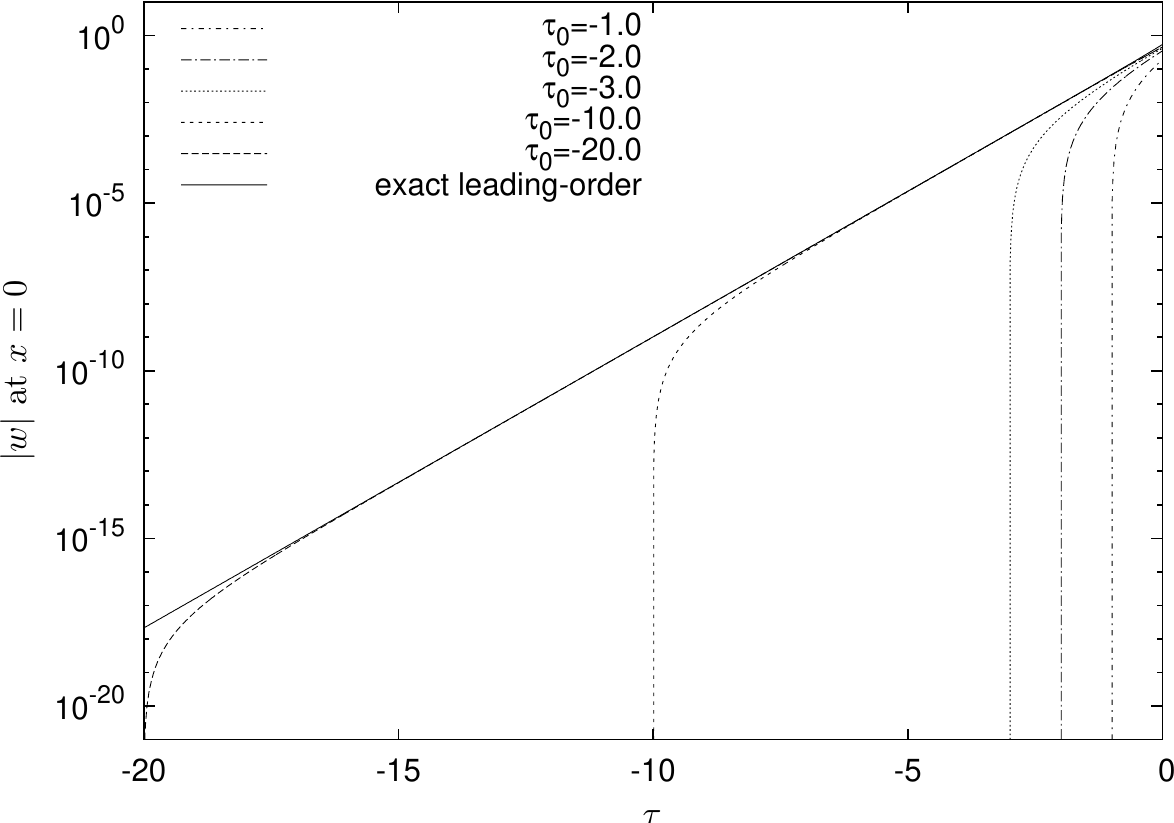}}\\
  \subfigure[$\lambda=1.90$.]{\includegraphics[width=0.49\linewidth]%
    {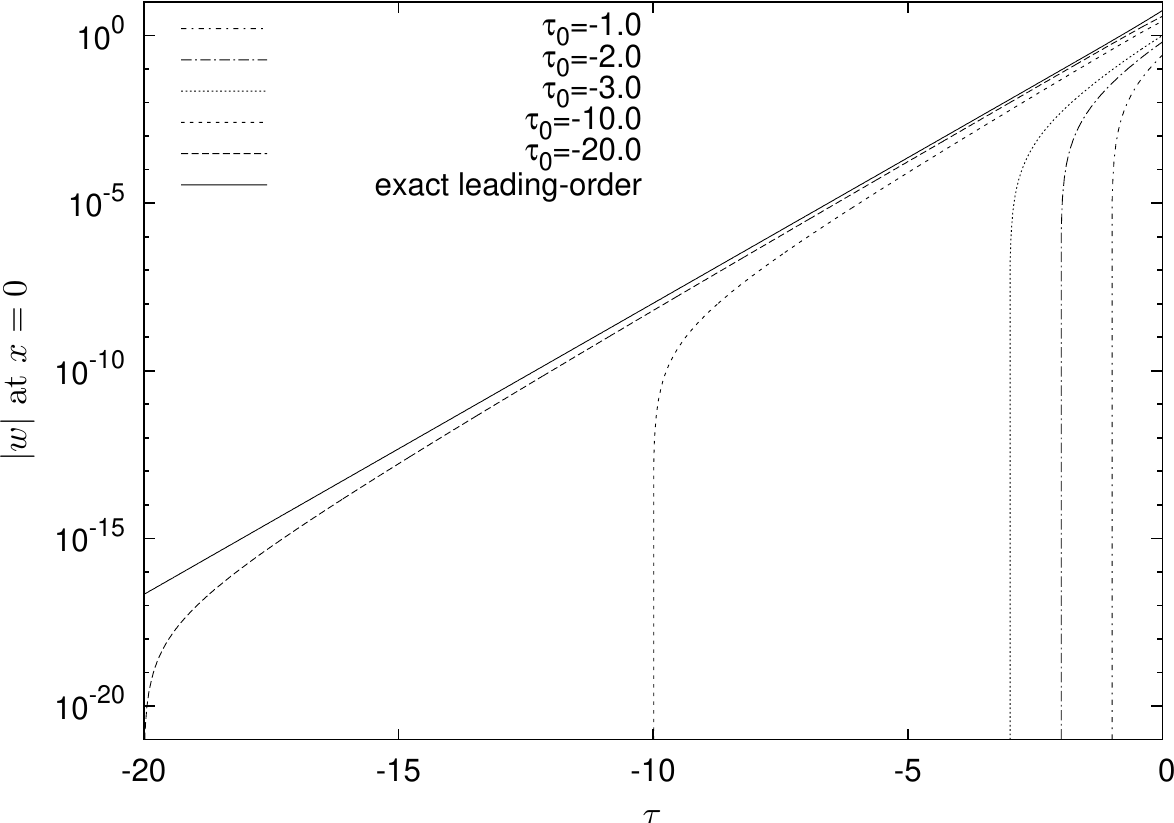}} 
  \subfigure[$\lambda=1.99$.]{\includegraphics[width=0.49\linewidth]%
    {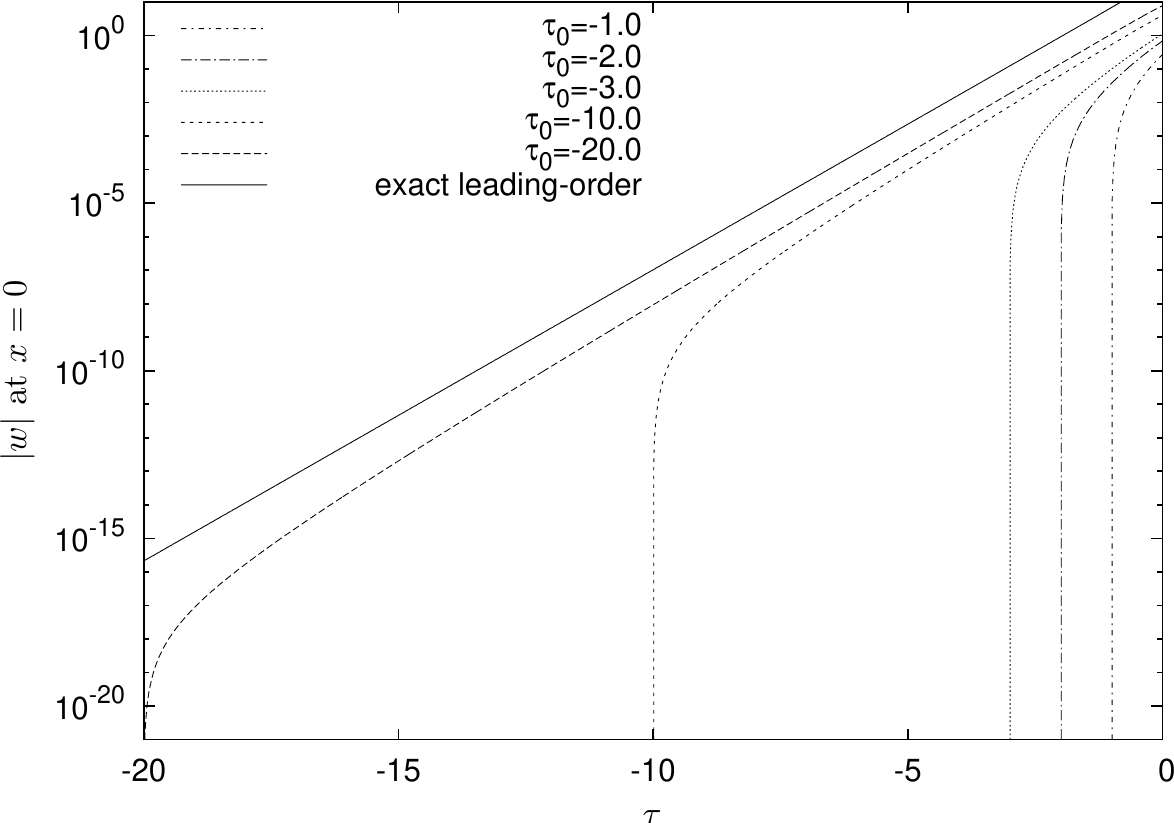}}      
  \caption{Numerical solutions Euler-Poisson-Darboux equation (as
    explained in the text).}
\label{fig:EPDNumerics}
\end{figure}

First we confirm that the numerical solutions converge in second-order
when $\Delta\tau$ and $\Delta x$ are changed proportionally to each
other for a given fixed choice of initial time $\tau_0>0$. In the
following, we choose the resolution so that discretization errors are
negligible relative to other errors. In \Figref{fig:EPDNumerics} now,
we show the following results obtained with $N=20$,
$\Delta\tau=0.003$. Here, $N$ is the number of spatial grid points,
i.e.~one has $\Delta x=2\pi/N$. The CFL-parameter is $\Delta t/\Delta
x\approx 0.01$, and we find that the runs are stable for all
$\tau<5$. For each of the plots of \Figref{fig:EPDNumerics}, we fix a
value of $\lambda$ and study the convergence of the approximate
solutions to the (leading-order of the) exact solution
\Eqref{eq:leadingorderexact} for various values of the initial time
$\tau_0$. We plot the value at one spatial point $x=0$ only.  The
convergence rate for $\tau_0\to-\infty$ is fast if $\lambda=1$ or
$\lambda=0.01$, but becomes lower, the more $\lambda$ approaches the
value $2$, where it becomes zero.  This is in exact agreement with our
expectations and consistent with the error estimates in
\Propref{prop:existenceSVIPlinear}.  Hence the numerical results are
very promising and confirm the analytic expectations.

Let us comment on numerical round-off errors. All numerical runs in
this paper were done with double-precision (binary64 of IEEE
754-2008), where the real numbers are accurate for $16$ decimal
digits. However, for the case $\tau_0=-20$ for instance, the second
spatial derivative of the unknown in the equation is multiplied by
$\exp(-40)\approx 10^{-18}$ at the initial time which is not resolved
numerically and hence could possibly lead to a significant
error. This, however, does not seem to be the case since we obtained
virtually the same numerical solution with quadruple precision
(binary128 of IEEE 754-2008), i.e.~when the numbers in the computer
are represented with $34$ significant decimal digits.


\section{Application to Gowdy spacetimes} 
\label{sec:Gowdy}
\subsection{Background material}

Let us provide some background material on Gowdy spacetimes
\cite{Gowdy73,Chrusciel}. Introduce coordinates $(t,x,y,z)$ such that
$(x,y,z)$ describe spatial sections diffeomorphic to $T^3$ while $t$
is a timelike variable.  We can arrange that the Killing fields
associated with the Gowdy symmetry coincide with the coordinate
vector fields $\del_y$, $\del_z$ in a global manner so that the
spacetime metric reads
\begin{equation*}
  g=\frac 1{\sqrt t} \, e^{\Lambda/2}(-dt^2+dx^2)+t \, (e^P(dy+Qdz)^2+e^{-P}dz), 
  \qquad t >0. 
\end{equation*}
Hence, the metric depends on three coefficients $P=P(t,x)$, $Q=Q(t,x)$,
and $\Lambda=\Lambda(t,x)$. We also assume spatial periodicity with
periodicity domain $U:=[0,2\pi)$.

In the chosen gauge, the Einstein's
vacuum equations imply the following second-order wave equations for $P,Q,$ 
\begin{equation}
  \label{eq:originalgowdy}
  \aligned
    P_{tt} + \frac{P_t}{t} - P_{xx} & = e^{2P} ( Q^2_t - Q^2_x),
    \\
    Q_{tt} + \frac{Q_t}{t} - Q_{xx} & = -2(P_t Q_t - P_x Q_x),
  \endaligned 
\end{equation}
which are decoupled from the wave equation satisfied by the third coefficient $\Lambda$:  
\begin{equation}
  \label{eq:evolLambda}
  \Lambda_{tt}-\Lambda_{xx}
  = P_x^2-P_t^2+e^{2P}(Q_x^2-Q_t^2).
\end{equation}
Moreover, the Einstein equations also imply constraint equations, which read 
\begin{subequations}
  \label{eq:constraints}
  \begin{align}
    \label{eq:constraints1}
    \Lambda_x&=2t \, \big( P_x P_t+ e^{2P}Q_x Q_t\big),
    \\
    \label{eq:constraints2}
    \Lambda_t&=t \, \big( P_x^2+t e^{2P}Q_x^2 + P_t^2 + e^{2P}Q_t^2 \big).
  \end{align}
\end{subequations}

It turns out that
\Eqref{eq:evolLambda} can sometimes be ignored in the following sense. 
Given a time $t_0>0$, 
we can prescribe initial data $(P,Q)|_{t_0}$ for
the system \eqref{eq:originalgowdy} while assuming the condition
\begin{equation*}
\int_0^{2\pi}(P_x P_t+e^{2P}Q_x Q_t) \, dx = 0 \qquad \text{ at } t=t_0.
\end{equation*}
Then, the first constraint \Eqref{eq:constraints1} determines the
function $\Lambda$ at the initial time, up to a constant which we
henceforth fix. Next, one easily checks that the solution $(P,Q)$ of
\eqref{eq:originalgowdy} corresponding to these initial data does
satisfy the compatibility condition associated with
\eqref{eq:constraints} and, hence, \eqref{eq:constraints} determines
$\Lambda$ uniquely for {\sl all} times of the evolution.  Moreover,
one checks that \eqref{eq:evolLambda} is satisfied identically by the
constructed solution $(P,Q,\Lambda)$.  One can also consider the
alternative viewpoint which follows from the natural $3+1$-splitting
and treats the three equations
\eqref{eq:originalgowdy}-\eqref{eq:evolLambda} as an evolution system
for the unknowns $(P,Q,\Lambda)$, and \eqref{eq:constraints} as
constraints that propagate if they hold on an initial hypersurface. In
any case, equations \eqref{eq:originalgowdy} represent the essential
set of Einstein's field equations for Gowdy spacetimes. We refer to
\eqref{eq:originalgowdy} as the \textit{Gowdy equations} and focus our
attention on them in most of what follows. An alternative, more
geometrical formulation of Einstein's field equation for Gowdy
symmetry has been introduced in \cite{Andersson2}.

\subsection{Heuristics about singular solutions of the Gowdy
  equations}
\label{sec21} 

We provide here a formal discussion which motivates the (rigorous) analysis in subsequent sections.  
Based on extensive numerical experiments \cite{BergerMoncrief,Berger97,Andersson2}, it was first conjectured 
(and later established rigorously \cite{Ringstrom6,Ringstrom7}) that as one approaches the singularity
the spatial derivative of solutions $(P,Q)$ to \eqref{eq:originalgowdy} becomes negligible and $(P,Q)$  
should approach a solution of the {\sl ordinary} differential equations 
\be
\aligned
& P_{tt} + {P_t \over t} = e^{2P} Q^2_t, 
\qquad 
 Q_{tt} + {Q_t \over t} = - 2 \, P_t \, Q_t. 
\endaligned
\label{GS.ODE}
\ee
These equations are referred to, in the literature\footnote{To our knowledge, this terminology has been introduced in \cite{Eardley71}
and, in the context of Gowdy spacetimes,
was first used in \cite{IsenbergMoncrief}.}, as the {\sl velocity term dominated} (VTD) equations. 
Interestingly enough, they admit solutions given explicitly by 
\be
\label{PQ} 
\aligned 
& P(t,x) = \log \big( \alpha \, t^{k} (1 + \zeta^2 t^{-2k})\big), 
\qquad
 Q(t,x) = \xi - {\zeta \, t^{-2k} \over \alpha \, (1 + \zeta^2 t^{-2k})},
\endaligned
\ee
where $x$ plays simply the role of a parameter and 
$\alpha>0, \zeta, \xi, k$ are arbitrary $2\pi$-periodic functions of $x$. 

Based on \eqref{PQ}, it is a simple matter to determine the first terms in the expansion of the function $P$ near $t=0$, 
that is
$$
\lim_{t \to 0} {P(t,x) \over \log t} = \lim_{t \to 0} t \, P_t(t,x) = -|k|, 
$$
hence
$$
\lim_{t \to 0} \big( P(t,x) + |k (x)| \log t \big) = \varphi(x), 
        \qquad
\varphi :=\begin{cases}
        \log\alpha,     &  k < 0, 
        \\
        \log (\alpha(1 + \zeta^2)),   & k = 0, 
        \\
        \log (\alpha\zeta^2),   & k > 0.
        \end{cases}
$$
Similarly, for the function $Q$ we obtain if $\zeta\not=0$, 
$$
\aligned
& \lim_{t \to 0} Q(t,x) = q (x), 
        \qquad 
q := \begin{cases}
        \xi,   & k < 0, 
        \\
         \xi - \dfrac{\zeta}{\alpha(1+\zeta^2)},  &  k = 0,
        \\
         \xi - \dfrac{1}{\alpha \zeta}, & k > 0,
        \end{cases}
\\
& \lim_{t \to 0} t^{-2|k|} \, \big( Q(t,x) - q (x) \big) = \psi(x), 
        \qquad \psi :=\begin{cases}
        -\dfrac{\zeta}{\alpha}, &  k < 0, 
        \\
        \hskip.41cm 0, & k = 0,
        \\
        \hskip.18cm \dfrac{1}{\alpha\zeta^3}, & k >0.
        \end{cases}
\endaligned 
$$
If $\zeta=0$, then $Q\equiv \zeta$.

From \eqref{PQ}, we thus have the expansion 
\be
\aligned
& P = -|k| \log t + \varphi + o(1),
\qquad \quad
 Q = q + t^{2|k|} \psi+o(t^{2|k|}), 
\endaligned
\label{GS.expansion}
\ee
in which $k, \varphi, q, \psi$ are functions of $x$. 
In general, $P$ {\sl blows-up to $+\infty$} when one approaches the singularity, while $Q$ 
remains {\sl bounded.}   
Observe that the sign of $k$ is irrelevant as far the asymptotic expansion is concerned, and we are allowed to restrict 
attention to $k\ge 0$.

By plugging the explicit solution into the nonlinear terms arising in \eqref{GS.ODE} one sees that 
$e^{2P} \, Q_t^2$ is of order $t^{2(|k|-1)}$ which is negligible since the left-hand side of the 
$P$-equation is of order $t^{-2}$, at least when $k \neq 0$. 
On the other hand, the nonlinear term $P_t Q_t$ is of order $t^{2(|k|-1)}$, which is the same order as the left-hand side
of the $Q$-equation. It is not negligible, but we observe that $P_t Q_t$ has the same behavior as $-(|k|/t) \, Q_t$. 

In fact, observe that the homogeneous system deduced from \eqref{GS.ODE}: 
\be
P_{tt} + {P_t \over t} = 0,
\qquad \quad
Q_{tt} + {1 - 2 k \over t} Q_t = 0,
\label{GS.ODE-simple}
\ee
is solved precisely by the leading-order terms in \Eqref{GS.expansion}. 
This tells us that, as $t \to 0$, 
 the term $e^{2P} Q^2_t$ is negligible in the first equation in 
 \Eqref{GS.ODE}, 
 while $P_tQ_t + (|k|/t)Q_t$ is negligible at $t=0$.
This discussion hence allows us to conclude that 
as far as the behavior at the coordinate singularity $t=0$ is concerned, 
the nonlinear VTD equations \eqref{GS.ODE} are well approximated by the system \eqref{GS.ODE-simple}.

\ 

We return now to the nonlinear terms which were not included in the VTD equations, 
but yet are present in the full model \eqref{eq:originalgowdy}. 
 Allowing ourselves to differentiate the expansion \eqref{GS.expansion}, we
  get the following leading-order terms at $t=0$:
$$
e^{2P} \, Q_x^2 = \begin{cases}
        t^{-2|k|} e^{2\varphi}\, q_x^2 + \ldots,    & q_x \neq 0, 
        \\
        2e^{2\varphi}|k|_x \psi\log t + \ldots,    & q_x=0,\quad |k|_x\not=0, 
        \\
        e^{2\varphi}\psi_x^2 + \ldots,      & q_x = 0,\quad |k|_x=0, \quad \psi_x \neq 0,  
        \end{cases} 
$$
$$
P_x \, Q_x 
=  \begin{cases}
-\log t \, |k|_x \, q_x + \ldots,   &  |k|_x, q_x \neq 0, 
\\
\varphi_x \, q_x + \ldots,   &  |k|_x = 0, \quad \varphi_x, q_x \neq 0, 
\\
-2\log^2 t\, t^{2|k|}\, (|k|_x)^2 \psi + \ldots,   &  |k|_x\neq 0,\quad q_x=0
\\
t^{2|k|} \varphi_x \, \psi_x + \ldots,   &  |k|_x= q_x = 0, \quad \varphi_x, \psi_x \neq 0. 
\end{cases} 
$$


To check (formally) the validity of the expansion \eqref{GS.expansion} we now return 
to the full system. Consider the nonlinear term $e^{2P} \, Q_x^2$ in \eqref{eq:originalgowdy},
and observe the following: 
\begin{itemize}
\item Case $q_x \neq 0$ everywhere on an open subinterval of $[0,2\pi]$. Then, on the one hand,  
the left-hand side of the first equation in  \eqref{eq:originalgowdy}
 is of order $t^{-2}$, at most. 
On the other hand, the term $e^{2P} \, Q_x^2$ is negligible with respect to $t^{-2}$ 
if and only if $|k| <1$ and is of the same order if $|k|=1$.  
\item Case $q_x = 0$ on an open subinterval of $[0,2\pi]$. Then,  
$e^{2P} \, Q_x^2$ is negligible with respect to $t^{-2}$, and no condition on the velocity 
$k$ is required on that interval. 
\item Case $q_x(x_0) = 0$ at some isolated point $x_0$. Then, no definite conclusion can be obtained 
and a ``competition'' between $|k|$ (which may approach the interval $[0,1]$) 
and $q_x$ (which approaches zero) is expected. 
\end{itemize} 

Similarly, at least when $|k|_x \, q_x \neq 0$, 
the nonlinear term $P_x Q_x$ is of order $\log t$ and, therefore, 
negligible with respect to $t^{2(|k|-1)}$ (given by the left-hand side of the second equation in \eqref{eq:originalgowdy})
if and only if $|k| \leq 1$. Points where $|k|_x$ or $q_x$
vanish lead to a less singular behavior and condition on the velocity can also be relaxed.

The formal derivation above strongly suggests that we seek solutions to the full nonlinear equations 
admitting an asymptotic expansion of the form \eqref{GS.expansion}, that is 
$$
P = - k \, \log t + \varphi + o(1),\quad
\qquad 
Q = q + t^{2 k} \, ( \psi + o(1)),  
$$
where $k \geq 0$ and $\varphi, q, \psi$ are prescribed. 
In other words, these solutions asymptotically approach a solution of the VTD equations and, in consequence, such solutions
will be referred to as {\sl asymptotically velocity term dominated} (AVTD) solutions \cite{IsenbergMoncrief}.

Based on this analysis and extensive numerical experiments, it has
been conjectured that asymptotically as one approaches the coordinate
singularity $t=0$ the function $P(t,x)/\log t$ should approach some
limit $k=k(x)$, referred to as the \emph{asymptotic velocity,} and
that $k(x)$ should belong to $[0,1)$ with the exception of a zero
measure set of ``exceptional values''. The reason for this name of $k$
is the following. Based on the work by Geroch \cite{Geroch1,Geroch2},
it was noted by Moncrief \cite{Moncrief} that the evolution equations
\eqref{eq:originalgowdy} for $P$ and $Q$ can be considered as wave map
equations with the hyperbolic space as the target space. If a solution
of these equations has an expansion of the form \eqref{GS.expansion}
at $t=0$, then the velocity of the image points of this map, which
must be defined with the correct convention of the sign and must be
measured with respect to the hyperbolic metric,
approaches $k$ as $t\rightarrow 0$.

It was demonstrated in \cite{Andersson2} that solutions of the Gowdy
equations which are compatible with \eqref{GS.expansion} approach
certain Kasner solutions at $t=0$, with possibly different parameters
along each timeline to the singularity.

\subsection{Gowdy equations as a second-order hyperbolic Fuchsian
  system}
\label{sec:gowdyequations2ndhypFuchs}

The first step in our (rigorous) analysis of the Gowdy equations
\Eqsref{eq:originalgowdy} now is to write them as a system of
second-order hyperbolic Fuchsian equations. After multiplication by
$t^2$, the equations \eqref{eq:originalgowdy} immediately take the
second-order hyperbolic Fuchsian form
\begin{align*}
  D^2P&=t^2\del_x^2 P+e^{2P} (DQ)^2 - t^2 e^{2P}(\del_x Q)^2,
  \\
D^2Q&=t^2\del_x^2 Q-2 DP DQ + 2t^2 \del_x P \del_x Q.
\end{align*} 
The general canonical two-term expansion then reads  
$$
P(t,x)=P_*(x)\log t+P_{**}(x)+\ldots 
$$
for the function $P$ and, similarly, an expansion $Q_*(x)\log
t+Q_{**}(x)+\ldots$ for the function $Q$ with prescribed data $Q_*,
Q_{**}$.  At this stage, we do not make precise statements about the
(higher-order) remainders, yet. In any case, the Fuchsian theory does
not apply to this system due to the presence of the term $-2 DP DQ$
(with the exception of the cases $P_*=0$ or $Q_*=0$).  Namely, this
term does not behave as a positive power of $t$ at $t=0$ when we
substitute $P$ and $Q$ by their canonical two-term expansions, but
this is required by the theory. The reason for this problem is the
significance of the nonlinear term in the source-term as found
in \Sectionref{sec21}, cf.\ \eqref{GS.ODE-simple}.

We propose to add the term $-2k DQ$ to the equation for $Q$ where $k$
is a prescribed (smooth, spatially periodic) function depending on
$x$, only. The function $k$ will play the role of the asymptotic
velocity mentioned before.  This yields the system of equations
\begin{equation}
  \label{eq:Fuchsiangowdy}
  \begin{split}
    D^2P&=t^2\del_x^2 P+e^{2P} (DQ)^2 - t^2 e^{2P}(\del_x Q)^2,\\
    D^2Q-2k DQ&=t^2\del_x^2 Q-2 (k+DP) DQ 
    + 2t^2 \del_x P \del_x Q.
  \end{split}
\end{equation}
The resulting system is of second-order hyperbolic
Fuchsian form with two equations, corresponding to 
\begin{equation*}
\lambda_1^{(1)}=\lambda_2^{(1)}=0,
\quad 
\lambda_1^{(2)}=0,
\quad 
\lambda_2^{(2)}=-2k.
\end{equation*}
Here, the
superscript determines the respective equation of the system
\Eqsref{eq:Fuchsiangowdy}.  If we assume that $k$ is a strictly
positive function, as we will do in all of what follows, the expected
leading-order behavior at $t=0$ given by the canonical two-term
expansions is
\begin{equation}
  \label{eq:leadingorderPQ}
  \aligned
  P(t,x)& =P_*(x)\log t+P_{**}(x)+\ldots,
  \\
  Q(t,x)& =Q_*(x)+Q_{**}(x)t^{2k(x)}+\ldots. 
\endaligned
\end{equation}
One checks easily that the problem associated with the term $-2 DP DQ$
before does not arise if $P_*=-k$. Indeed, the canonical two-term
expansion \Eqref{eq:leadingorderPQ} is consistent with the heuristics
of the Gowdy equations above and we recover the singular initial value
problem studied rigorously in \cite{KichenassamyRendall,Rendall00} and
numerically in \cite{ABL}. We only mention here without further notice
that the case $k\equiv 0$ with the logarithmic canonical two-term
expansion for $Q$ is covered by the following discussion. Furthermore,
the case of $k$ vanishing at only certain points may be also included
via a suitable renormalization of the asymptotic data, see
\Eqref{eq:renormalizeddata}.

When $P_*=-k$, the function $k$ plays a two-fold role in
\Eqref{eq:Fuchsiangowdy}. On the one hand, it is an asymptotic data
for the function $P$ and, on the other hand, it is a coefficient of
the principal part of the second equation. In order to keep these two
roles of $k$ separated in a first stage, we consider the system
\begin{equation}  
  \label{eq:FuchsiangowdyManip}
  \begin{split}
    D^2P&=t^2\del_x^2 P+e^{2P} (DQ)^2 - t^2 e^{2P}(\del_x Q)^2,\\
    D^2Q-2k DQ&=t^2\del_x^2 Q-2 (-P_*+DP) DQ 
    + 2t^2 \del_x P \del_x Q,
  \end{split}
\end{equation}
instead of \Eqref{eq:Fuchsiangowdy}. Studying the singular initial
value problem with two-term asymptotic data means that we search for
solutions to \Eqref{eq:FuchsiangowdyManip} of the form (as $t \to 0$)
\begin{equation}
  \label{eq:expansionGowdy}
  \begin{split}
    P(t,x)&=P_*(x)\log t+P_{**}(x)+w^{(1)}(t,x), \\
    Q(t,x)&=Q_*(x)+Q_{**}(x)t^{2k(x)}+w^{(2)}(t,x),
  \end{split}
\end{equation}
for general asymptotic data $P_*$, $P_{**}$, $Q_*$, $Q_{**}$, and
remainders $w^{(1)}$, $w^{(2)}$. After studying the well-posedness for
this problem, we can always choose $P_*$ to coincide with $-k$ and,
therefore, recover our original Gowdy problem
\eqref{eq:Fuchsiangowdy}-\eqref{eq:leadingorderPQ}. For simplicity in
the presentation, we always assume that $k$ is a $C^\infty$ function.

In the following discussion, we write the vector-valued remainder as
$w:=(w^{(1)},w^{(2)})$, and we fix some asymptotic data $P_*$, $P_{**}$, $Q_*$, and
$Q_{**}$ and choose the leading-order term $u$ according to \Eqref{eq:expansionGowdy}. The source-term operator 
$F[w](t,x)=:(F_1[w](t,x),F_2[w](t,x))$  
reads
\begin{align*}
  F_1[w]=&\left(t^{P_*} e^{P_{**}} e^{w^{(1)}} \big( 2k\, t^{2k}Q_{**}+Dw^{(2)} \big) \right)^2
  \\
  &-\left(t^{P_*} e^{P_{**}} e^{w^{(1)}} \, \big( t \, \del_x Q_*
    +2\del_x k t^{2k}t\log t Q_{**}
    + t^{2k}t\del_xQ_{**}+t\del_xw^{(2)} \big) \right)^2,
\end{align*}
and
$$
\aligned
F_2&[w]
\\
= &-2 Dw^{(1)} \, \big( 2k\, t^{2k}Q_{**}+Dw^{(2)} \big)
\\
 & +2 \Big( t \del_xP_*\log t+t\del_x P_{**}+t\del_xw^{(1)} \Big)
     \Big( \big(t \del_x Q_* + 2 \, \del_x k t^{2k}t\log t Q_{**}
  + t^{2k}t\del_xQ_{**}+t\del_xw^{(2)} \Big).
\endaligned
$$

\subsection{Properties of the source-term operator}

To establish the well-posedness of the singular initial value
problem for the Gowdy equations,
 we need first to derive certain decay properties of the source-term
operator $F$ consistent with \Sectionref{sec:existence}.
Let us introduce some notation specific to the Gowdy equations.  Let
$X_{\delta,\alpha_1,k}^{(1)}$ be the space defined as above based on
the coefficients of the first equation in
\Eqsref{eq:FuchsiangowdyManip} and, similarly, let
$X_{\delta,\alpha_2,k}^{(2)}$ be the space associated with the second
equation. By definition, a vector-valued map $w= (w^{(1)}, w^{(2)})$
belongs to $X_{\delta,\alpha,k}$ precisely if $w^{(1)}\in
X_{\delta,\alpha_1,k}^{(1)}$ and $w^{(2)}\in
X_{\delta,\alpha_2,k}^{(2)}$, with $\alpha:=(\alpha_1,\alpha_2)$.  An
analogous notation is used for the spaces
$\Xt_{\delta,\alpha_1,k}^{(1)}$, $\Xt_{\delta,\alpha_2,k}^{(2)}$ and
$\Xt_{\delta,\alpha,k}$.

Now we are ready to state a first result about the source-term of \Eqsref{eq:FuchsiangowdyManip}. 

\begin{lemma}[Operator $F$ in the finite differentiability class]
  \label{lem:finitediffF}
Fix any $\delta>0$ and any asymptotic data 
$P_*, P_{**},Q_*,Q_{**}\in H^m(U)$, $m\ge 2$. 
Suppose there exist $\eps>0$ and a continuous function
  $\alpha=(\alpha_1,\alpha_2) : U \to (0,\infty)^2$ so that, at each $x\in U$, 
  \begin{subequations}
    \label{eq:finitediffF}
    \begin{gather}
      \label{eq:finitediffF1}
      \alpha_1(x)+\eps<
      \min\big( 2(P_*(x)+2k(x)),2(P_*(x)+1)\big),
      \\
      \label{eq:finitediffF2}
      \alpha_2(x)+\eps<2(1-k(x)),\\
      \label{eq:finitediffF3}
      \alpha_1(x)-\alpha_2(x)>\eps+\min\big( 0,2k(x)-1\big),
      \\
      \eps<1.
    \end{gather}
  \end{subequations}
  Then, the operator $F$ associated with the system \eqref{eq:FuchsiangowdyManip} and the 
  given asymptotic data 
  maps 
  $\Xt_{\delta,\alpha,m}$ into $X_{\delta,\alpha+\eps,m-1}$ 
and satisfies the
  following Lipschitz continuity condition: For each $r>0$ and for some 
  constant $C>0$ (independent of $\delta$), 
$$
E_{\delta,\alpha+\eps,m-1}\Big[ F[w]-F[\what] \Big](t)
  \le C \, \Et_{\delta,\alpha,m}[w-\what](t), \qquad  t\in (0,\delta]
$$
  for all $w,\what\in B_r\subset
  \Xt_{\delta,\alphat,m}$, where $B_r$ denotes the closed ball centered at the origin. 
\end{lemma}

In this lemma, since $P_*\in H^1(U)$, in particular, a standard Sobolev
inequality implies that $P_*$ can be identified with a unique bounded
continuous periodic function on $U$, and the inequality \eqref{eq:finitediffF1} makes sense pointwise.

\begin{proof} Consider the expression of $F$ given before. 
  Let $w\in\Xt_{\delta,\alpha,m}$ for some (so far unspecified)
  positive spatially dependent functions $\alpha_1,\alpha_2$, hence
  $w^{(1)}\in\Xt_{\delta,\alpha_1,m}^{(1)}$ and $w^{(2)}\in\Xt_{\delta,\alpha_2,m}^{(2)}$. 
  By a standard Sobolev inequality (since $m \geq 2$ and the 
  spatial dimension is $1$), 
  we get that $F[w](t,\cdot)\in H^{m-1}(U)$ for all
  $t\in(0,\delta]$.  Namely, if $m\ge 2$ we can control the nonlinear
  terms of $F[w](t,\cdot)$ in all generality for a given $t>0$ if any
  factor in any term of $F[w](t,\cdot)$, after applying up to $m-1$
  spatial derivatives, is an element in $L^\infty(U)$ -- with the
  exception of the $m$th spatial derivative of $w$ which is only
  required to be in $L^2(U)$. This is guaranteed by the Sobolev
  inequalities.  Having found that $F[w](t,\cdot)\in H^{m-1}(U)$ for
  all $t\in(0,\delta]$, it is easy to check that $F_1[w]\in
  X_{\delta,\alpha_1+\eps,0}^{(1)}$ if
  \begin{equation}
    \label{eq:firstconditionPre}
    \alpha_1(x)+\eps\le
    \min\Big( 2(P_*(x)+2k(x)),2(P_*(x)+1)\Big), \qquad  x\in U.
  \end{equation}
  Even more, condition \Eqref{eq:firstconditionPre} implies that
  $D^lF_1[w]\in X_{\delta,\alpha_1+\eps,0}^{(1)}$ for all $l\le
  m-1$. 
  
  Considering now spatial derivatives, we have to deal
  with two difficulties. The first one is that  
  logarithmic terms arise with each spatial derivative. We find 
  $\del_x^kD^lF_1[w]\in X_{\delta,\alpha_1+\eps,0}^{(1)}$ for
  all $l\le m-1$ and $k\le m-2$ and $k+l\le m-1$ (excluding first the case $k=m-1$, $l=0$)
   provided
  \begin{equation}
    \label{eq:firstcondition}
    \alpha_1(x)+\eps<
    \min\Big( 2(P_*(x)+2k(x)),2(P_*(x)+1)\Big),\qquad x\in U.
  \end{equation}

A second difficulty arises in the case $k=m-1$, $l=0$. Namely, since 
  $w\in\Xt_{\delta,\alpha,m}$ (and not in $X_{\delta,\alpha,m}$),
  it follows that in particular $t\del^m_x w^{(2)}\sim
  t^{2k+\alpha_2}$ (and not $t^{1+2k+\alpha_2}$); note that the
  function $\beta$ which determines the behavior of the characteristic
  speeds at $t=0$ is identically zero in the
  case of the Gowdy equations. The potentially problematic term is
  hence of the form $AB$ with 
  \begin{align*}
    A:= &t^{P_*} e^{P_{**}} e^{w^{(1)}}(t\del_x Q_*
    +2\del_x k t^{2k}t\log t Q_{**}
    + t^{2k}t\del_xQ_{**}+t\del_xw^{(2)}),
    \\
    B:=&t^{P_*} e^{P_{**}} e^{w^{(1)}}\bigl(\del_x^{m-1}(t\del_x Q_*
    +2\del_x k t^{2k}t\log t Q_{**}
    + t^{2k}t\del_xQ_{**})+t\del_x^{m}w^{(2)}\bigr),
  \end{align*}
  originating from taking $m-1$ spatial derivatives of $F_1[w]$. 
  To ensure $\del_x^{m-1}F_1[w]\in
  X_{\delta,\alpha_1+\eps,0}^{(1)}$, we need 
  \begin{equation}
    \label{eq:secondcondition}
    \alpha_1(x)+\eps< (P_*(x)+1)+(P_*(x)+2k(x)+\alpha_2(x)),
    \qquad x\in U.
  \end{equation}
  If \Eqref{eq:firstcondition} is satisfied, we have (for all $x$) 
$$
  \aligned
    \alpha_1(x)+\eps
    & <
    \min\Big( 2(P_*(x)+2k(x)),2(P_*(x)+1)\Big) 
 \\
    & \le (P_*(x)+1)+(P_*(x)+2k(x))
\endaligned
$$
  and, thus, \Eqref{eq:secondcondition} follows from
  \Eqref{eq:firstcondition}. In conclusion, 
  \Eqref{eq:firstcondition} is sufficient to guarantee that $F_1[w]\in
  X_{\delta,\alpha_1+\eps,m-1}^{(1)}$. 
  
  Let us proceed next with the analysis of the term 
  $F_2[w]$. If
  \begin{equation}
    \label{eq:thirdcondition}
    \alpha_1(x)-\alpha_2(x)\ge\eps,
    \quad\qquad 
    \alpha_2(x)+\eps<2(1-k(x)),\qquad x\in U,
  \end{equation}
  then $F_2[w]\in X_{\delta,\alpha_2+\eps,0}^{(2)}$. This
  inequality also implies that all time derivatives are in
  $X_{\delta,\alpha_2+\eps,0}^{(2)}$ as before. We have to deal
  with the same two difficulties as before when we consider spatial
  derivatives of $F_2[w]$. On the one hand, 
  equality in
  \Eqref{eq:thirdcondition} cannot occur due to additional
  logarithmic terms. On the other hand, we must be careful with the $(m-1)$-th
  spatial derivative of $F_2[w]$.  Here, the two problematic terms 
  are of the form $AB$ with either 
  \begin{align*}
   A:= &\del_x^{m-1}(t\del_xP_*\log t+t\del_x P_{**})
    +t\del_x^m w^{(1)},
    \\
    B:= &t\del_x Q_*
    +2\del_x k t^{2k}t\log t Q_{**}
    + t^{2k}t\del_xQ_{**}+t\del_xw^{(2)},
  \end{align*}
or else 
  \begin{align*}
    A:= & t\del_xP_*\log t+t\del_x P_{**}+t\del_x w^{(1)},
    \\
    B:= &\del_x^{m-1}(t\del_x Q_*
    +2\del_x k t^{2k}t\log t Q_{**}
    + t^{2k}t\del_xQ_{**})+t\del_x^mw^{(2)}.
  \end{align*}
  The first one is under control provided
  $\alpha_1(x)+1>2k(x)+\alpha_2(x)+\eps$, for all $x\in U$, while
  for the second one it is sufficient to require $\eps<1$.  The
  claimed Lipschitz continuity condition follows from the above arguments.
\end{proof}

Positive functions $\alpha_1$ and $\alpha_2$ and constants $\eps>0$
satisfying the hypothesis of \Lemref{lem:finitediffF} can obviously
exist only if $k(x)<1$ for all $x\in U$ (due to
\Eqref{eq:finitediffF2}). In \Lemref{444} below we identify a special
case where this limitation is avoided.  Hence, we make the assumption
that $0<k(x)<1$ for all $x$, which is consistent with our formal
analysis in
\Sectionref{sec21}.  As a consistency check for the case of interest
$P_*=-k$, let us determine under which conditions the inequalities
\Eqsref{eq:finitediffF} can be hoped to be satisfied at all. For this,
consider \Eqsref{eq:finitediffF1} and \eqref{eq:finitediffF3} in the
borderline case $\alpha_2=\eps=0$. This leads to the condition
$0<k<3/4$, which shows that \Lemref{lem:finitediffF} does not apply
within the full interval $0<k<1$. It is interesting to note that
Rendall was led to the same restriction in \cite{Rendall00}, but its
origin stayed unclear in his approach. Here, we find that this is
caused by the presence of the condition \eqref{eq:finitediffF3} in
particular which reflects the fact that $w$ is an element of the space
$\Xt_{\delta,\alpha,m}$ rather than of the smaller space
$X_{\delta,\alpha,m}$. Interestingly, we can eliminate this condition
and, hence, retain the full interval $0<k<1$, when we consider the
$C^\infty$-case, instead of finite differentiability. See also
\cite{Andersson2} for a detailed discussion of the different intervals
of $k$.

\begin{lemma}[Operator $F$ in the $C^\infty$ class. General theory]
  \label{lem:CinftyF}
  Fix any $\delta>0$ and any asymptotic data 
$P_*, P_{**},Q_*,Q_{**}\in C^\infty(U).
$
  Suppose there exist a constant $\eps>0$ and a continuous
  functions $\alpha=(\alpha_1,\alpha_2) :U\to (0,\infty)^2$ such  that, 
  at each $x\in U$, 
  \begin{subequations}
    \begin{gather}
      \label{eq:CinftyF1}
      \alpha_1(x)+\eps<
      \min\big( 2(P_*(x)+2k(x)),2(P_*(x)+1)\big),
      \\
      \label{eq:CinftyF2}
      \alpha_2(x)+\eps<2(1-k(x)),\\
      \label{eq:CinftyF3}
      \alpha_1(x)-\alpha_2(x)>\eps.
    \end{gather}
  \end{subequations}
  Then, for each integer $m \ge 1$, the operator $F$ maps 
  $X_{\delta,\alpha,m}$ into $X_{\delta,\alpha+\eps,m-1}$
   and satisfies the
  following Lipschitz continuity property:
  for each $r>0$ and some constant $C>0$ (independent of $\delta$), 
$$
  E_{\delta,\alpha+\eps,m-1}\big[ F[w] - F[\what]\big](t)
  \leq
   C \Et_{\delta,\alpha,m}[w-\what](t), \qquad t\in (0,\delta], 
$$ 
  for all $w,\what\in B_r \cap
  X_{\delta,\alpha + \eps,m}\subset\Xt_{\delta,\alpha+ \eps,m}$. 
\end{lemma}

The proof is completely analogous to that of
\Lemref{lem:finitediffF}. Since only spaces $X_{\delta,\alpha,k}$ need
to be checked (i.e.\ without the tilde) in the $C^\infty$-case we
obtain stronger control than in the finite differentiability case.  In
consequence, the $C^\infty$-case does not require the condition
\eqref{eq:finitediffF3}. Thus $k$ can have values in the whole
interval $(0,1)$ as we show in detail later.  In a special case, which
will be of interest for the later discussion, however, we can relax
the constraints for $k$ even in the finite differentiability case.

\begin{lemma}[Operator $F$ in the finite differentiability class. A special case]
\label{444} 
Fix any $\delta>0$ and any asymptotic data 
$P_*, P_{**},Q_{**}\in H^m(U)$, $Q_*=const$, $m \geq 2$. 
Suppose there exist 
  $\eps>0$ and 	a continuous function
  $\alpha=(\alpha_1,\alpha_2) : U \to (0,\infty)^2$ such that, at each $x\in U$, 
  \begin{align*}
    &\alpha_1(x)+\eps<2(P_*(x)+2k(x)),
    \\
    &\alpha_2(x)+\eps<2,
\qquad \quad 
    \alpha_1(x)-\alpha_2(x)>\eps-1, 
\qquad \qquad \eps<1.
  \end{align*} 
  Then, the operator $F$ satisfies the conclusions of \Lemref{lem:finitediffF}.
\end{lemma}

In the special case $Q_*=const$, we have hence characterized the map
$F$ for $k$ being any positive function in the finite
differentiability case.  The analogous result for the $C^\infty$-case
can also be derived.


\subsection{Well-posedness theory}

Relying on \Theoremref{th:well-posednessSIVP} and the results in the
previous sections, we now determine conditions that ensure that the
singular initial value problem for the Gowdy equations is well-posed.
Besides the properties of the source-operator $F$ already discussed,
we have to check the positivity of the energy dissipation matrix.
This leads us to the matrix
\begin{equation*}
  N^{(1)} :=
  \begin{pmatrix}
    \alpha_1 & -\eta/2 & 0 \\
    -\eta/2 
    & \alpha_1 
    & 0 \\
    0 & 0 
    &  \alpha_1-1 
  \end{pmatrix} 
\end{equation*}
for the first component and to the matrix 
\begin{equation*}
  N^{(2)}  :=
  \begin{pmatrix}
    2 k + \alpha_2 & -\eta/2 & 0 \\
    -\eta/2 
    & \alpha_2 
    & - 2 \del_x k (t \log t) \\
    0 & - 2 \del_x k (t \log t) 
    & 2 k +\alpha_2-1 
  \end{pmatrix} 
\end{equation*}
for the second component.  For the matrix $ N^{(1)}$ to be positive,
it is necessary that $\alpha_1(x)>1$ for all $x\in U$. However, if
$P_*=-k$, then condition \eqref{eq:finitediffF1} in
\Lemref{lem:finitediffF} in the finite differentiability case (or the
corresponding one in \Lemref{lem:CinftyF} in the $C^\infty$-case)
implies that $\alpha_1(x)<1$. Hence, in the same way as in Rendall
\cite{Rendall00}, one does not arrive at a well-posedness result for
the singular initial value problem yet.  However, since the positivity
of the energy dissipation matrix is the only part of the hypothesis in
\Theoremref{th:well-posednessSIVP} which is is violated, we can use
instead \Theoremref{th:WellPosednessHigherOrderSIVP} to prove
well-posedness of the singular initial value problem with asymptotic
solutions of sufficiently high order $j$.

Let us be specific about what we mean by $j$ being ``sufficiently
large'', and we now make some choice for the parameters $\alpha_1$,
$\alpha_2$ and $\eps$, consistent with \Lemref{lem:CinftyF}, which
will allow us to estimate the required size of $j$. We make no
particular effort to choose these quantities optimally, but still the
goal is to choose $j$ ``reasonably'' small. Henceforth, we restrict to
the $C^\infty$-case and $P_*(x)=-k(x)$ with $0<k(x)<1$ for all $x\in
U$.  We introduce positive constants $\mu_1$ and $\mu_2$ (with further
restrictions later) and the function $\chi(x):=1-2|x-1/2|$. The
condition \eqref{eq:CinftyF1} states that we must choose $\alpha_1(x)$
and $\eps$ so that $\alpha_1(x)+\eps<\chi(k(x))$. We set
\begin{equation}
  \label{eq:choicealpha1}
  \alpha_1(x):=1-\sqrt{4(k(x)-1/2)^2+\mu_1^2},
\end{equation}
and find $\chi(k(x))-\alpha_1(x)>\sqrt{1+\mu_1^2}-1$ for all $x\in
U$, provided $0<k(x)<1$. Similarly, we set
\begin{equation}
  \label{eq:choicealpha2}
  \alpha_2(x):=1-\sqrt{4(k(x)-1/2)^2+\mu_2^2},
\end{equation}
and it follows that
$\alpha_1(x)-\alpha_2(x)>\sqrt{1+\mu_2^2}-\sqrt{1+\mu_1^2}$ for
$\mu_2>\mu_1$. For the conditions \eqref{eq:CinftyF1}
and \eqref{eq:CinftyF3} to hold true, we have to choose
\[0<\mu_1<\mu_2,\quad\text{and}\quad
0<\eps\le\min\Biggl(
\sqrt{1+\mu_1^2}-1,\sqrt{1+\mu_2^2}-\sqrt{1+\mu_1^2}
\Biggr).
\] 
Condition \eqref{eq:CinftyF2} is then satisfied automatically. 

Now, assume in what follows that $k(x)\in (1/2-\Delta
k,1/2+\Delta k)$ for all $x\in U$ for a constant $\Delta k\in
(0,1/2)$. Then it is clear that both functions $\alpha_1$ and
$\alpha_2$ are positive for all such $k(x)$ if and only if
\[\mu_1<\mu_2<\sqrt{1-4(\Delta k)^2}.\]
This assumption will be made in the following. In
\Theoremref{th:WellPosednessHigherOrderSIVP}, we could choose $j$ as
small as possible if we pick the maximal allowed value for
$\eps$. Hence, we set
\[\eps:=\min\Biggl(\sqrt{1+\mu_1^2}-1,\sqrt{1+\mu_2^2}-\sqrt{1+\mu_1^2}
\Biggr).
\]
We find easily that
\[\sqrt{1+\mu_1^2}-1\le \sqrt{1+\mu_2^2}-\sqrt{1+\mu_1^2},\]
provided
\[\mu_1^2\le \frac 14 (\mu_2^2+2\sqrt{1+\mu_2^2}-2),\]
and check that this is consistent with the condition $0<\mu_1<\mu_2$
made before. In order to make a specific choice, we assume this
inequality for $\mu_1$ and hence obtain that
\begin{equation}
  \label{eq:choiceepsilon}
  \eps=\sqrt{1+\mu_1^2}-1.
\end{equation}
Now, in order to make the energy dissipation matrix positive, we must
choose $j$ so that for all $x\in U$,
$$
\aligned
\alphat_1(x) & :=\alpha_1(x)+(j-2)\kappa\eps>1,
\\
\alphat_2(x) & :=\alpha_2(x)+(j-2)\kappa\eps> 1 - 2k(x);
\endaligned
$$
cf.\ \Theoremref{th:WellPosednessHigherOrderSIVP}. These two
inequalities are satisfied for all functions $k$ under our assumptions
if in particular
\begin{equation}
  \label{eq:estimatej}
  j>2+\frac{\sqrt{4(\Delta k)^2+\mu_2^2}}
  {\kappa(\sqrt{1+\mu_1^2}-1)}.
\end{equation}
In any case, we choose the maximal value for $\mu_1$
\begin{equation}
  \label{eq:choicemu1}
  \mu_1:=\frac 12 \sqrt{\mu_2^2+2\sqrt{1+\mu_2^2}-2}, 
\end{equation}
since this minimizes the value on the right side of
\eqref{eq:estimatej}.  We find that for this value of $\mu_1$, the
right side of \eqref{eq:estimatej} is monotonically decreasing in
$\mu_2$ and diverges to $+\infty$ for $\mu_2\to 0$ for all
values of $\Delta k$.

\begin{theorem}[Well-posedness theory for the Gowdy equations] 
  \label{th:well-posedness}
  Consider some asymptotic data $ P_*=-k, \, P_{**}, \, Q_*, \, Q_{**}
  \in C^\infty(U), $ where $k$ is a smooth function $U\to (1/2-\Delta
  k,1/2+\Delta k)$ for a constant $\Delta k\in (0,1/2)$. Then, the
  singular initial value problem with asymptotic solutions of order
  $j$ has a unique solution with remainder $w\in
  X_{\delta,\alpha+(j-2)\kappa\eps,\infty}$ for some sufficiently
  small $\delta>0$ and some $\kappa<1$. Here, the exponents
  $\alpha=(\alpha_1,\alpha_2)$ and $\eps$ are given in
  \Eqsref{eq:choicealpha1}, \eqref{eq:choicealpha2}, and
  \eqref{eq:choiceepsilon} explicitly in terms of the data and
  parameters $\mu_1, \mu_2$ chosen such that $\mu_1$ is an explicit
  expression in $\mu_2$ given in \Eqref{eq:choicemu1} while $\mu_2$ is
  a sufficiently close to (but smaller than) $\sqrt{1-4(\Delta k)^2}$,
  and the order of differentiation $j$ satisfies
  $$
  j > 2+\frac{2}{\sqrt{3-4(\Delta k)^2+2\sqrt{2-4(\Delta k)^2}}-2}.
  $$
\end{theorem}

The above condition implies that to reach $\Delta k\to 0$ we need $j>7$, while
$\Delta k\to 1/2$ requires $j\to\infty$. Although our
estimates may not be quite optimal, the latter implication cannot be avoided.


\subsection{Fuchsian analysis for the function 
  \texorpdfstring{$\Lambda$}{Lambda}} 

So far we have considered the equations \Eqref{eq:originalgowdy} for
$P$ and $Q$. We can henceforth assume that these equations are solved
identically for all $t>0$ (and $t\le\delta$ for some $\delta>0$) and
that hence $P$ and $Q$ are given functions with leading-order behavior
\Eqref{eq:leadingorderPQ} and remainders in a given
$X_{\delta,\alpha,k}$. The equations which remain to be solved in
order to obtain a solution of the full Einstein's field equations are
\Eqref{eq:evolLambda} and \Eqref{eq:constraints}. In particular we are
interested in the function $\Lambda$ in order to obtain the full
geometrical information. We compute $\Lambda$ from a singular
initial value problem with ``data'' on the singularity analogously to
$P$ and $Q$. The following discussion resembles the previous one and
we only discuss new aspects now.

Clearly, the three remaining equations \Eqref{eq:evolLambda} and
\eqref{eq:constraints} for $\Lambda$ are overdetermined, and hence
solutions will exist only under certain conditions. Let us define the
following ``constraint quantities'' from \Eqref{eq:constraints}
\begin{gather*}
  C_1(t,x):=-\del_t\Lambda+t(P_x)^2+e^{2P}t(Q_x)^2+t(\del_tP)^2
  +e^{2P}t(\del_tQ)^2,\\
  C_2(t,x):=-\Lambda_x+2P_x DP+2e^{2P}Q_x DQ.
\end{gather*}
Moreover, we define
\[H(t,x):=-\Lambda_{tt}+\Lambda_{xx}+ P_x^2-P_t^2+e^{2P}(Q_x^2-Q_t^2)
\]
from \Eqref{eq:evolLambda}. From the evolution equations for $P$ and
$Q$, we find the subsidiary system
\begin{equation}
  \label{eq:constrprop}
  \del_t C_1=\del_x C_2+H,\quad \qquad \del_t C_2=\del_x C_1.
\end{equation}
These equations have the following consequences. Suppose that we use
\Eqref{eq:constraints2} as an evolution equation for $\Lambda$. This
implies that $C_1\equiv 0$ for all $t>0$. Moreover, suppose that we
prescribe data at some $t_0>0$ (indeed $t_0$ is allowed to be zero
later) so that $C_2(t_0,x)=0$ for all $x\in U$. Then the equations
imply that $H\equiv 0$ and $C_2\equiv 0$ for all $t>0$ and thus we
have constructed a solution of the full set of field
equations. Alternatively, let us use \Eqref{eq:evolLambda} as the
evolution equation for $\Lambda$, i.e.\ $H\equiv0$. Suppose that we
prescribe data so that $C_1(t_0,x)=C_2(t_0,x)=0$ at some $t_0$. It
follows that $C_1\equiv C_2\equiv 0$ for all $t>0$ because the
evolution system \Eqref{eq:constrprop} for $C_1$ and $C_2$ is
symmetric hyperbolic. Again, Einstein's field equations are
solved. 

Now, we want to consider the case $t_0=0$. First note that
\Eqsref{eq:constrprop} is regular even at $t=0$.  Suppose that $P$ and
$Q$ are functions with leading-order behavior
\Eqref{eq:leadingorderPQ} and remainders in a given
$X_{\delta,\alpha,k}$ with $k\ge 1$. If there exists a function $w_3$
so that
\begin{equation}
  \label{eq:expansionLambda}
  \Lambda(t,x)=\Lambda_*(x)\log t+\Lambda_{**}(x)+w_3(t,x)
\end{equation}
with $w_3$ converging to zero in a suitable norm at $t=0$ and
\begin{equation}
  \label{eq:asymptconstr}
  \Lambda_*(x)=k^2(x),\quad
  \Lambda_{**}(x)=\Lambda_0
  +2\int_0^x k(\tilde x)(-\del_{\tilde x}P_{**}(\tilde x) +2
  e^{2P_{**}(\tilde x)}Q_{**}(\tilde x)\del_{\tilde x} Q_*(\tilde x))\,d\tilde x,
\end{equation}
where $\Lambda_0$ is an arbitrary real constant,
then, in particular, 
\[\lim_{t\rightarrow 0} C_2=0.\]
It also follows that $\lim_{t\rightarrow 0} t C_1=0$.  Let us first
use \Eqref{eq:constraints2} as a singular evolution equation for
$\Lambda$. Since this is ``only'' a singular ODE, one can show easily
that there exists a unique solution for $\Lambda$ for $t>0$ which
obeys the two-term expansion above, and hence $C_1\equiv 0$. Our
discussion before implies that $H,C_2\equiv 0$. Hence we obtain a
solution of the full Einstein's field equation for all
$t>0$. Alternative, choose \Eqref{eq:evolLambda} as the evolution
equation for $\Lambda$ now. This equation can be written in
second-order hyperbolic Fuchsian form
\[D^2\Lambda-t^2\del_x^2\Lambda
=(t\del_xP)^2+(D\Lambda-(DP)^2)+e^{2P}((t\del_x
  Q)^2-(DQ)^2).
\] 
Indeed this equation is compatible with the leading-order expansion
\Eqref{eq:expansionLambda} at $t=0$ and we can show well-posed of this
singular initial value problem in the same way as we did for the
functions $P$ and $Q$ before (also going to sufficiently high-order in
$j$). In particular, for any asymptotic data $\Lambda_{*}$ and
$\Lambda_{**}$, not necessarily those given by
\Eqref{eq:asymptconstr}, there exists a unique solution of this
equation $\Lambda$ with remainder $w_3$ in a certain space
$X_{\delta,\alpha,k}$. By means of uniqueness we find that the
solution $\Lambda$ of this equation coincides with the solution for
$\Lambda$ obtained using \Eqref{eq:constraints2} as the evolution
equation. Hence, we must have $H,C_1,C_2\equiv 0$ for all
$t>0$, and thus also this method yields a solution of the full
Einstein's field equations.

Note that periodicity and \Eqref{eq:asymptconstr} implies that the
asymptotic data for $P$ and $Q$ must satisfy the relation
\[\int_0^{2\pi} k(\tilde x)(-\del_{\tilde x}P_{**}(\tilde x) +2
e^{2P_{**}(\tilde x)}Q_{**}(\tilde x)\del_{\tilde x} Q_*(\tilde
x))\,d\tilde x=0
\]
for smooth solutions.


\section{Numerical experiments}
\label{sec:numericalsolGowdy}

\subsection{Test~1. Homogeneous pseudo-polarized solutions}

We continue our discussion with the singular initial value problem for
the Gowdy equations. In all of what follows we consider the singular
initial value problem with two-term asymptotic data for the Gowdy
equations. As we show this works very well and we get good
convergence. This is a strong indication that the standard singular
initial value problem is well-posed. In contrast, recall from
\Theoremref{th:well-posedness} that our analytical techniques are only
sufficient to show that the initial value problem with asymptotic
solutions of sufficiently high order is well-posed for the Gowdy
equations.

Before we proceed with ``interesting'' solutions of the Gowdy equations,
let us start with a case for which we can construct an explicit
solution and hence test the numerical implementation.  Let $\Pt$ and
$\Qt$ be solutions of the polarized equations in the homogeneous case,
i.e.\ set $\Qt=0$ and $\Pt(t,x)=\Pt(t)$. In this case, it follows
directly that the exact solution of the Gowdy equations is
\[\Pt(t)=-k\log t+\Pt_{**},\]
where both $k$ and $\Pt_{**}$ are arbitrary constants. The
corresponding full solutions of Einstein's equations are Kasner
solutions whose parameters\footnote{In the conventions of
  \cite{Wainwright}, we have $p_1=(k^2-1)/(k^2+3)$,
  $p_2=2(1-k)/(k^2+3)$, $p_3=2(1+k)/(k^2+3)$, and 
  the three flat cases are
  realized by $k=1$, $k=-1$ and $|k|\rightarrow\infty$.} are
determined by $k$ exclusively ($\Pt_{**}$ is just a gauge
quantity). By a reparametrization of the Killing orbits of the form
\[\tilde x_2=x_2/\sqrt 2+x_3/\sqrt 2,\quad 
\tilde x_3=-x_2/\sqrt 2+x_3/\sqrt 2,
\] 
where $\tilde x_2$, and $\tilde x_3$ are the coordinates used to
represent the orbits of the polarized solution above, the same
solution gets re-expressed in terms of functions
\be
\label{hhhh}
P=\log\cosh (-k\log t+\Pt_{**}),\quad Q=\tanh (-k\log t+\Pt_{**}).
\ee
These functions $(P,Q)$ are again solutions of
\eqref{eq:originalgowdy}. Asymptotically at $t=0$, they satisfy 
\begin{equation*} 
  P=-k\log t+(\Pt_{**}-\log 2)+\ldots,\quad Q=1-2e^{-2\Pt_{**}}t^{2k}+\ldots,
\end{equation*}
from which we can read off the corresponding asymptotic data. 

\begin{figure}[t]  
 \centering
 \subfigure[$k=0.5$.]{%
   \includegraphics[width=0.49\textwidth]{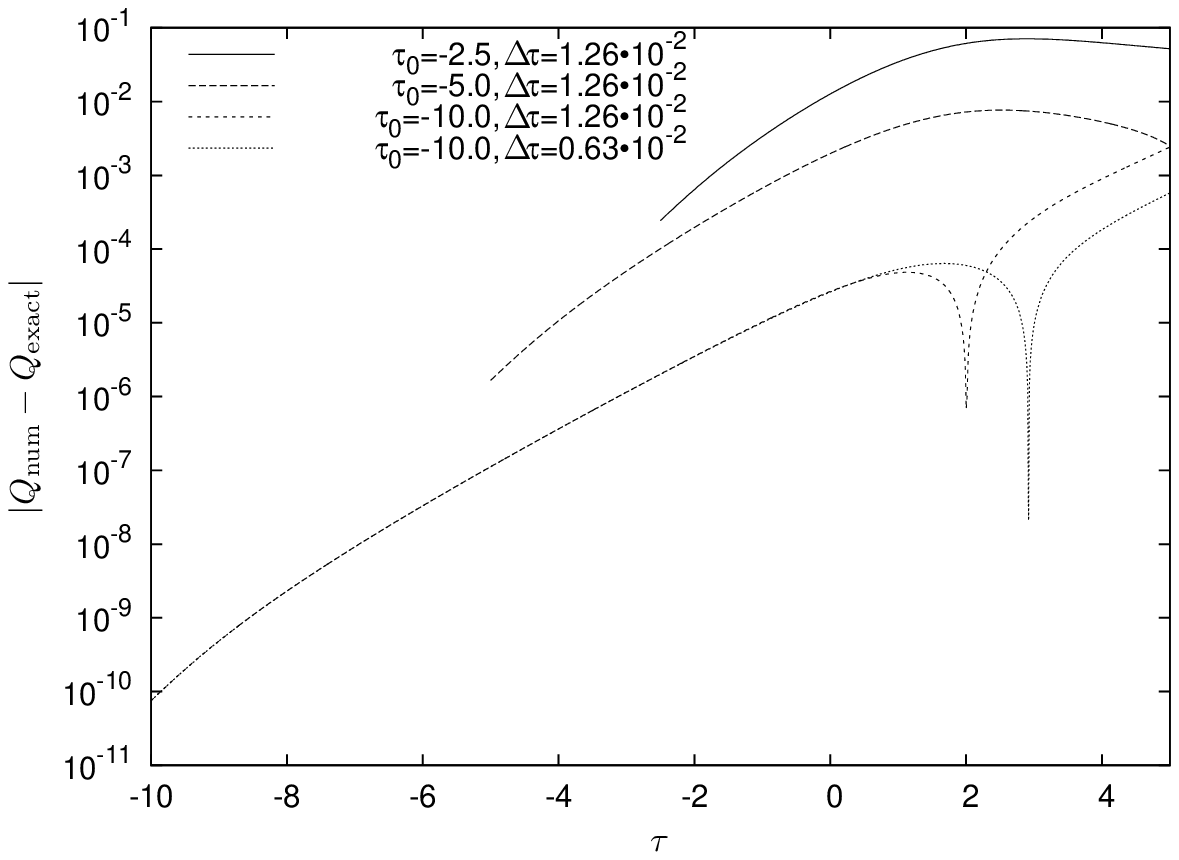}}
 \subfigure[$k=0.9$.]{%
   \includegraphics[width=0.49\textwidth]{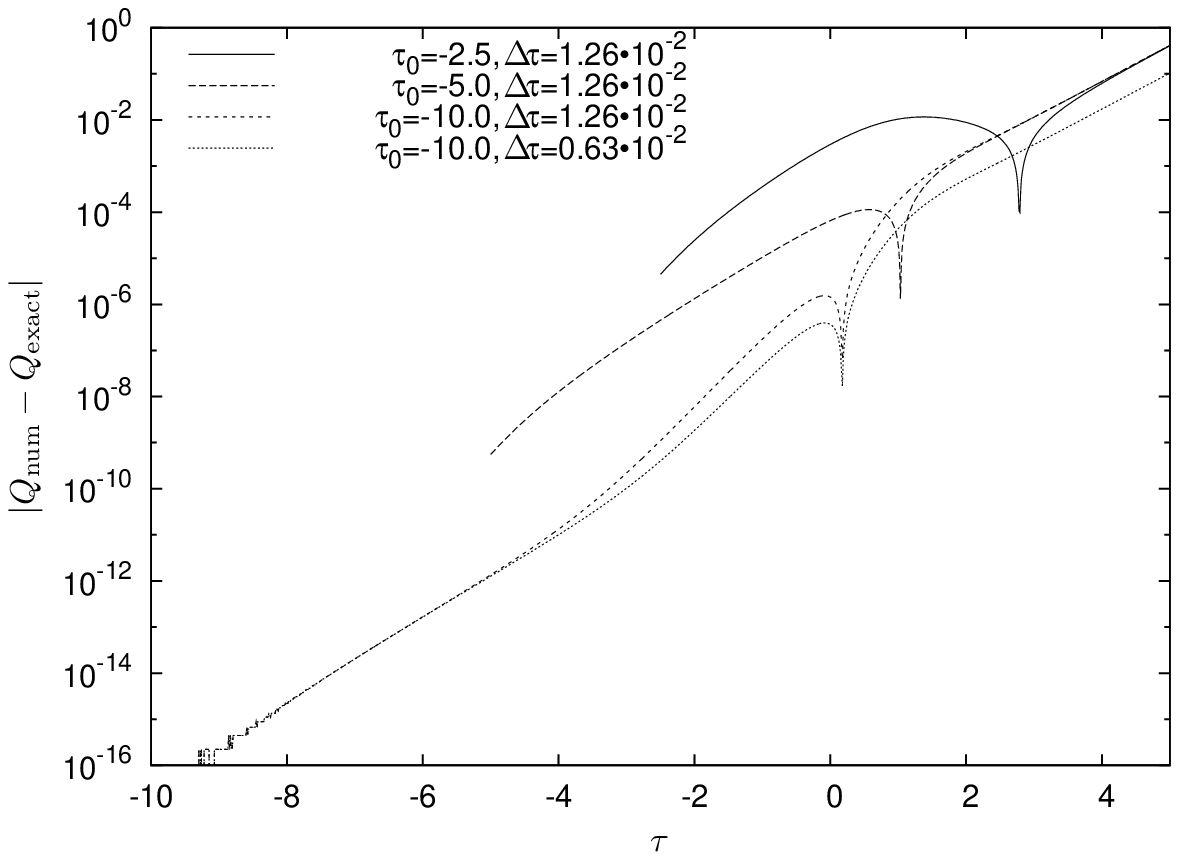}}
 \caption{Convergence of numerical solutions of the test case 1 (as
   explained in the text).}
 \label{fig:GowdyPseudoPolHomog}
\end{figure}
Now we compute the solutions corresponding to these asymptotic data
numerically and compare them to the exact solution \Eqref{hhhh}. We
pick $\Pt_{**}=1$, so that $P_{**}=1-\log 2$, $Q_{*}=1$ and $Q_{**}=-2
e^{-2}$. Since the solution is spatially homogeneous -- in fact this
is an ODE problem -- we only need to do the comparison at one spatial
point. 

The results are presented in \Figref{fig:GowdyPseudoPolHomog}
where we plot the difference of the numerical and the exact value of
$Q$ versus time for various values of $\tau_0$. In the first plot,
this is done for $k=0.5$ and in the second plot for $k=0.9$. The plots
confirm nice convergence of the approximate solutions to the exact
solution. The fact that each approximate solution diverges from the
exact solution almost exponentially in time is a feature of the
approximate solutions themselves and not of the numerical
discretization, as is checked by comparing two different values of
$\Delta\tau$ in these plots.  From our experience with the
Euler-Poisson-Darboux equation, we could have expected that the
convergence rate is lower in the case $k=0.9$ than in the case $k=0.5$
(note that $k$ plays the same role $\lambda/2$). In the case of the
Euler-Poisson-Darboux equation, the rate of convergence decreases when
$\lambda$ approaches $2$, due to the influence of the second-spatial
derivative term in the equation. In the spatially homogeneous case
here, however, this term is zero and hence this phenomenon is not
present. The ``spikes'' in \Figref{fig:GowdyPseudoPolHomog} are just a
consequence of the logarithmic scale of the horizontal axes and the
fact that the numerical and exact solutions equal for some instances
of time.

\subsection{Test~2. General Gowdy equations}
\begin{figure}[t]  
 \centering
 \subfigure[$A=0.2$.]{%
   \includegraphics[width=0.49\textwidth]{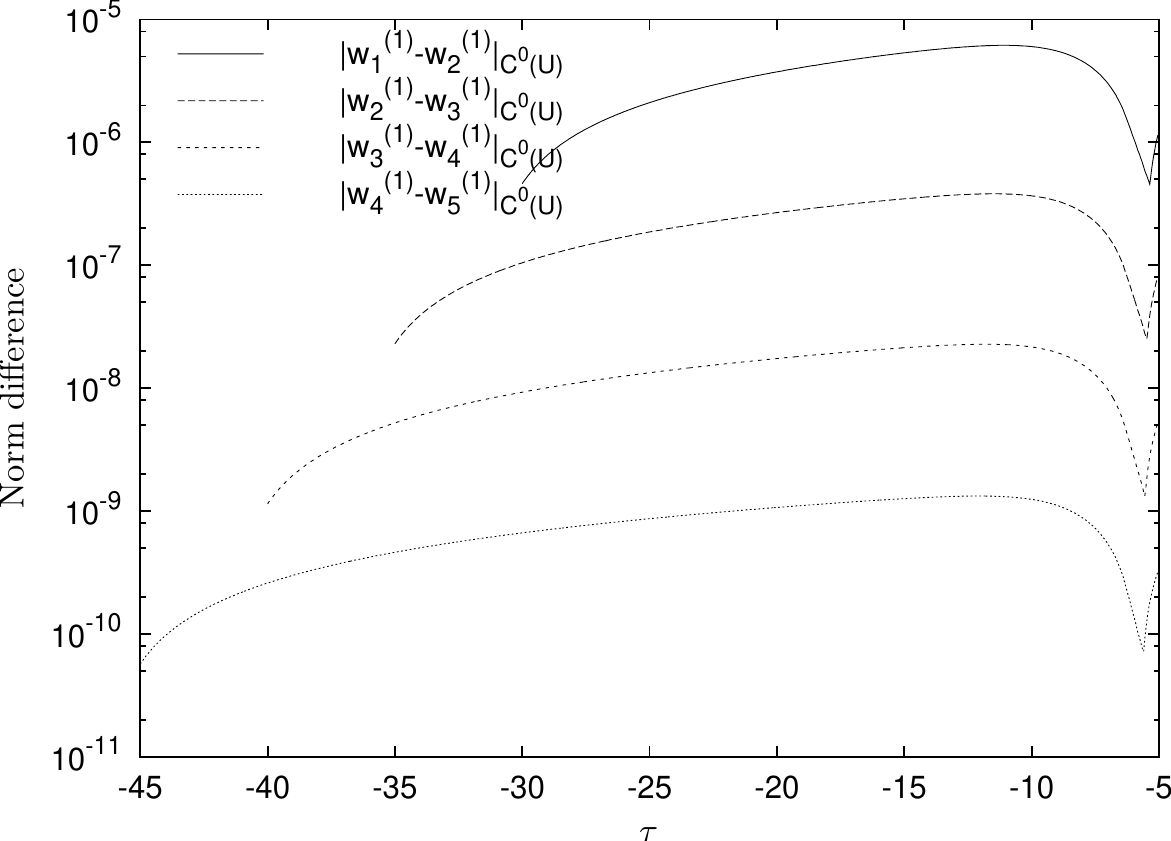}}
 \subfigure[$A=0.4$.]{%
   \includegraphics[width=0.49\textwidth]{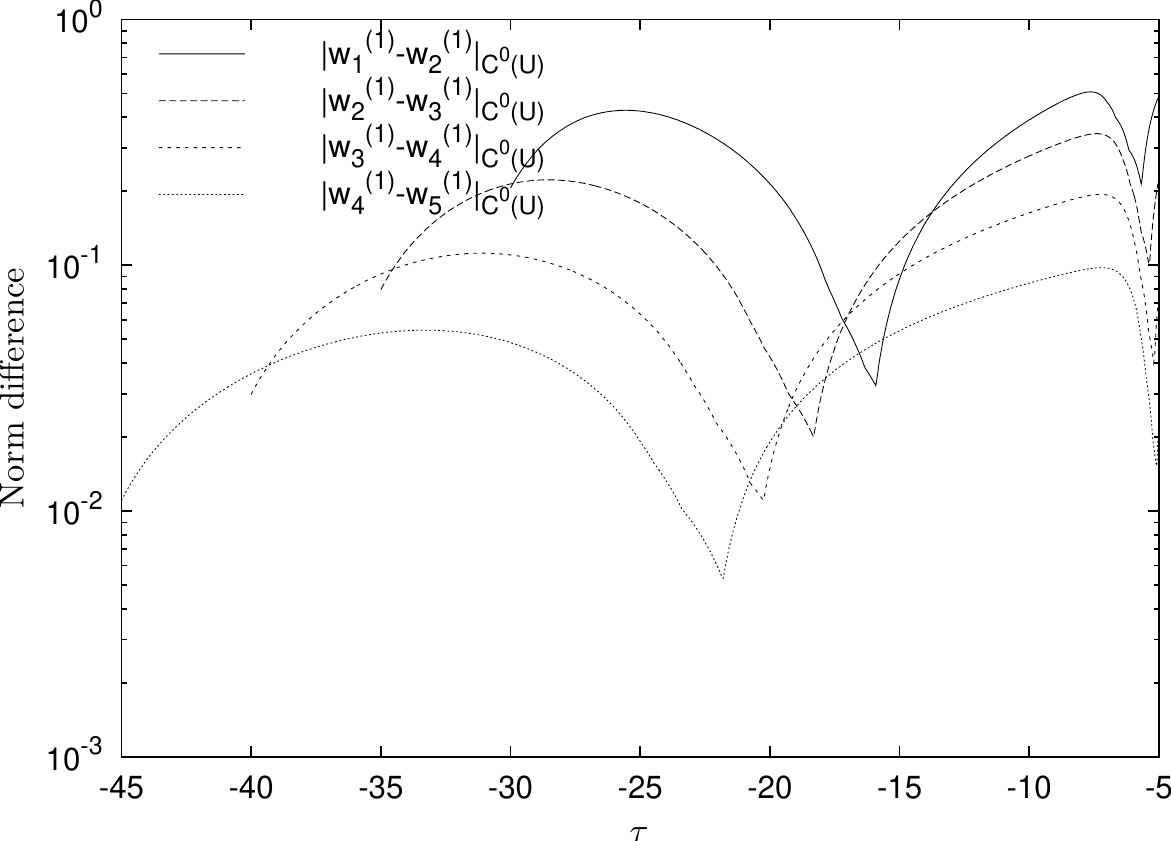}}
 \caption{Convergence of numerical solutions of the test case 2 as
   explained in the text.}
 \label{fig:GowdyGeneralConvergence}
\end{figure}
Now we want to study the convergence for a ``generic'' inhomogeneous
Gowdy case (still ignoring the equation for the quantity
$\Lambda$). Here we choose the following asymptotic data
$$
\aligned
& k(x)=1/2+A\cos(x),\qquad Q_*=1.0+\sin(x),
\\
& P_{**}=1-\log 2+\cos(x),\qquad Q_{**}=-2 e^{-2},
\endaligned
$$  
with a constant $A\in (-1/2,1/2)$. We do not know of an explicit
solution in this case. In \Figref{fig:GowdyGeneralConvergence}, we
show the following numerical results for $A=0.2$ and $A=0.4$,
respectively. For the given value of $A$, we compute five approximate
solutions with initial times $\tau_0=-30,-35,-40,-45,-50$ numerically,
each with the same resolution $\Delta\tau=0.01$ and $N=80$. The
resolution parameters have been chosen so that the numerical
discretization errors are negligible in the plots of
\Figref{fig:GowdyGeneralConvergence}. Then, for each time step for
$\tau\ge-30$, we compute the supremum norm in space of the difference
of the remainders $w^{(1)}$ of the two approximate solutions given by
$\tau_0=-30$ and $\tau_0=-35$. In this way we obtain the first curve
in each of the plots of \Figref{fig:GowdyGeneralConvergence}. The same
is done for the difference between the cases $\tau_0=-35$ and
$\tau_0=-40$ for all $\tau\ge-35$ to obtain the second curve
etc. Hence these curves yield a measure of the convergence rate of the
approximation scheme (without referring to the exact solution).  In
agreement with our observation for the Euler-Poisson-Darboux equation,
the convergence rate is high if $k$ is close to $1/2$ and becomes
lower, the more $k$ touches the ``extreme'' values $k=0$ and $k=1$.

Much in the same way as for the Euler-Poisson-Darboux equation we find
that double precision is sufficient for these computations despite of
the fact that $\exp(2\tau)$ is $10^{-44}$ for $\tau=-50$.


\subsection{Test 3. Gowdy spacetimes containing a Cauchy horizon}
\label{sec:GowdyCauchyhorizon}

The papers
\cite{ChruscielIsenbergMoncrief,ChruscielIsenberg,ChruscielLake,HennigAnsorg,IsenbergMoncrief}
were devoted to the construction and characterization of Gowdy
solutions with Cauchy horizons in order to prove the strong cosmic
censorship conjecture in this class of spacetimes.  Spacetimes with
Cauchy horizons are expected to have saddle and physically
``undesired'' properties, in particular they often allow various
inequivalent smooth extensions. This has the undesired consequence
that the Cauchy problem of Einstein's field equations does not select
one of them uniquely.  Some explicit examples are known, but most of
the analysis is on the level of existence proofs and asymptotic
expansions.

Hence, it is of interest to construct such solutions numerically and
analyze them in much greater detail than possible with purely analytic
methods. Constructing these solutions numerically, however, is
delicate since the strong cosmic censorship conjecture suggests that
they are instable under generic perturbations. It can hence often be
expected that numerical errors would most likely ``destroy the Cauchy
horizon''. This is so, in particular, when the singular time at $t=0$
is approached backwards in time from some regular Cauchy surface at
$t>0$, i.e.~for the ``backward approach''.

In the Gowdy case, where the strong cosmic censorship conjecture has
been proven \cite{Ringstrom7}, however, there are clear criteria for
the asymptotic data so that the corresponding solution of the singular
initial value problem has a Cauchy horizon (or only pieces thereof;
cf.~below) at $t=0$, as discussed in \cite{ChruscielIsenbergMoncrief}
for the polarized case and in \cite{ChruscielLake} for the general
case. Our novel method here allows us to construct such solutions with
arbitrary accuracy and it can hence be expected that this allows us to
study the saddle properties of such solutions.  Our main aim so far is
to compute such a solution and hence to demonstrate the feasibility of
our approach. A follow-up work will be devoted to the numerical
construction and detailed analysis of relevant classes of such
solutions.

Motivated by the results in \cite{ChruscielIsenbergMoncrief}, we choose
the asymptotic data as follows
\begin{align*}
  &k(x)=
  \begin{cases}
    1, & x\in [\pi,2\pi],\\
    1-e^{-1/x}e^{-1/(\pi-x)}, & x\in (0,\pi),
  \end{cases}&\quad &P_{**}(x)=1/2,\\
  &Q_*(x)=0,&\quad &Q_{**}(x)=\begin{cases}
    0, & x\in [\pi,2\pi],\\
    e^{-1/x}e^{-1/(\pi-x)}, & x\in (0,\pi),
  \end{cases}\\
  &\Lambda_*(x)=k^2(x),&\quad
  &\Lambda_{**}(x)=2.
\end{align*}
With these asymptotic data, the corresponding solution has a smooth
Cauchy horizon at $(t,x)\in \{0\}\times (\pi,2\pi)$ (namely where
$k\equiv 1$), and a curvature singularity at $(t,x)\in \{0\}\times
(0,\pi)$ (namely where $0<k<1$). Note that the function $k$ is smooth
everywhere (but not analytic). Our analysis
in \Sectionref{sec:gowdyequations2ndhypFuchs} shows that we are allowed to
set $k=1$ at some points since $\del_xQ_*=0$. This motivates our
choice of $Q_*$.  With this, our choice of $Q_{**}$ implies that the
solution is polarized on the ``domain of dependence''\footnote{The
  notion of ``domain of dependence'' for the singular initial value
  problem follows from the energy estimate.}
of the ``initial data'' interval $(\pi,2\pi)$. All data were chosen as
simple as possible to be consistent with the constraints.

\begin{figure}[t]  
 \centering  
 \includegraphics[width=0.49\textwidth]{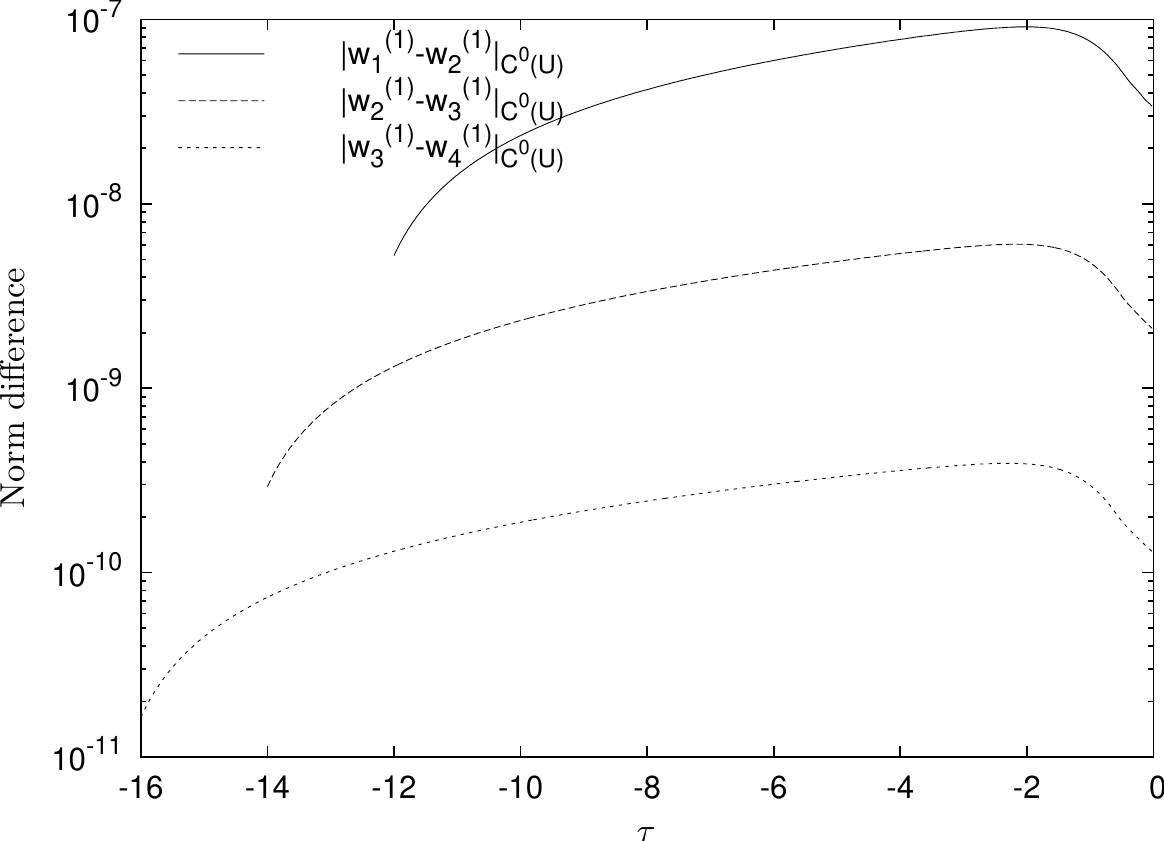}
 \caption{Convergence of numerical solutions corresponding to
   asymptotic data in \Sectionref{sec:GowdyCauchyhorizon}.}
 \label{fig:CHConvergence}
\end{figure}
First we repeated the same error analysis as for the previous Gowdy
case, see~\Figref{fig:CHConvergence}. For all the runs in the plots,
we choose $N=500$, $\Delta\tau=0.005$ which guarantees that
discretization errors are negligible in the plot. We find that our
numerical method allows us to compute the Gowdy solution very
accurately. Here, we solve the full system for
$(P,Q,\Lambda)$.

\begin{figure}[t] 
 \centering
 \subfigure[Kretschmann scalar at $\tau=-10.0$.]{%
   \includegraphics[width=0.49\textwidth]{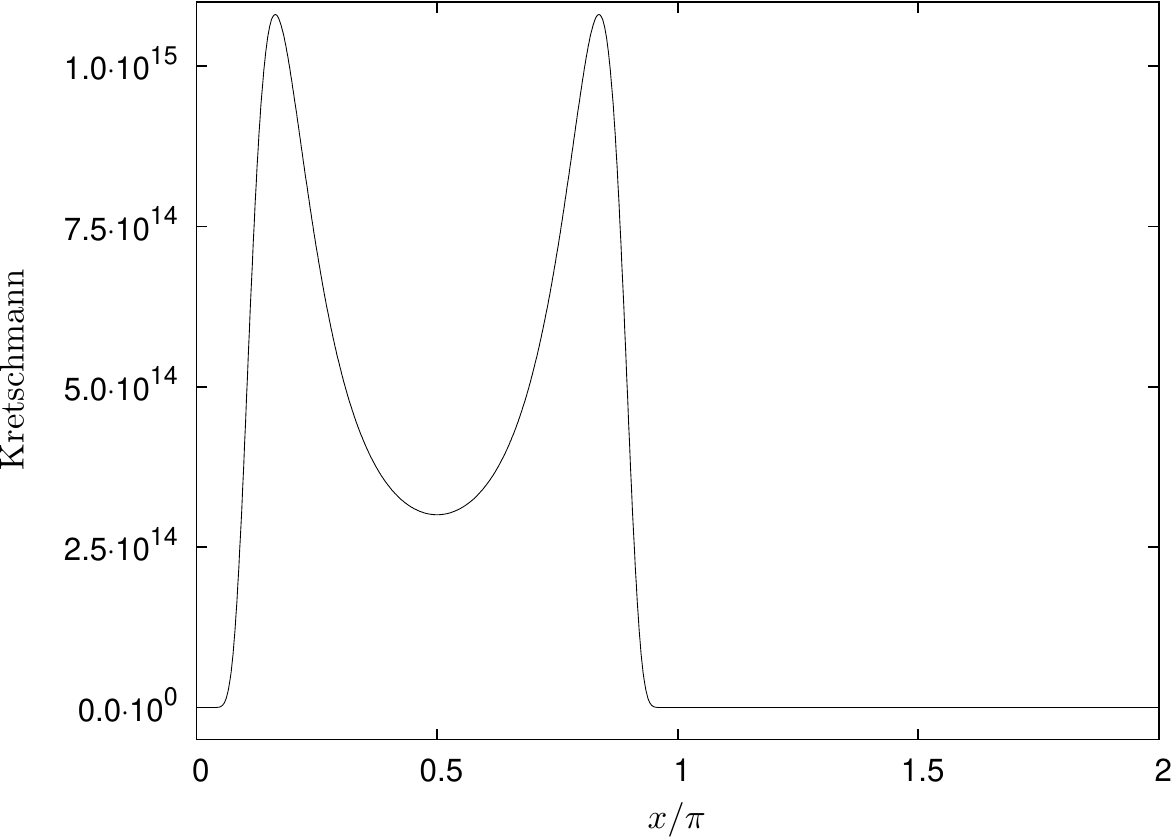}}
 \subfigure[Kretschmann scalar at $\tau=0.0$.]{%
   \includegraphics[width=0.49\textwidth]{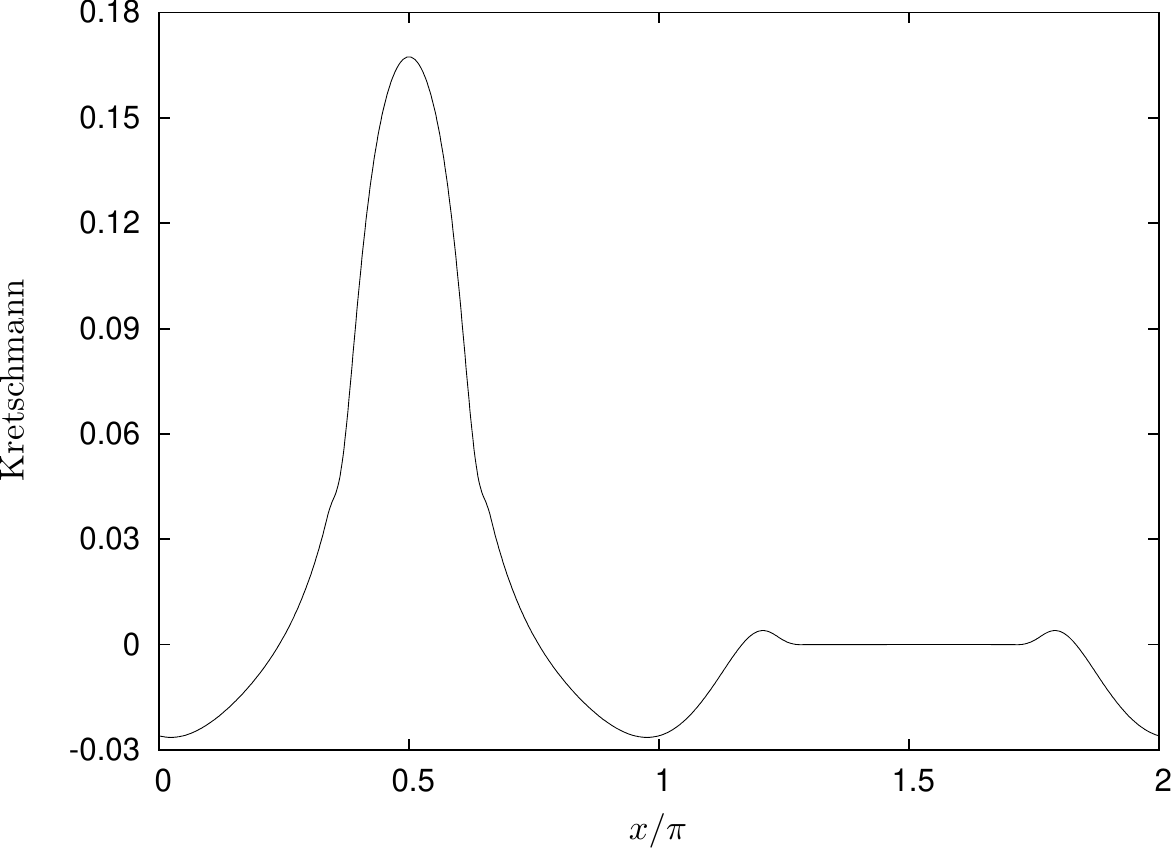}}

 \subfigure[Remainders of $P,Q$ at $\tau=-10.0$.]{%
   \includegraphics[width=0.49\textwidth]{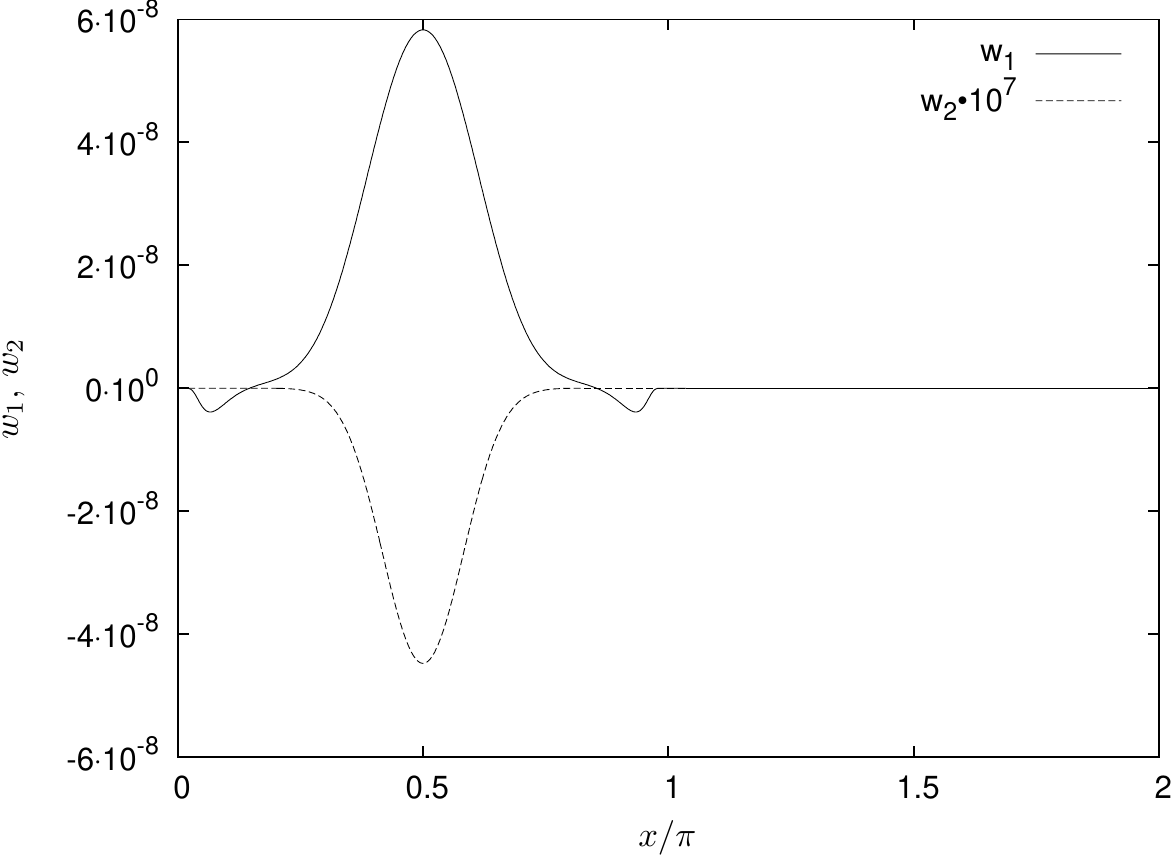}}
 \subfigure[Remainders of $P,Q$ at $\tau=0.0$.]{%
   \includegraphics[width=0.49\textwidth]{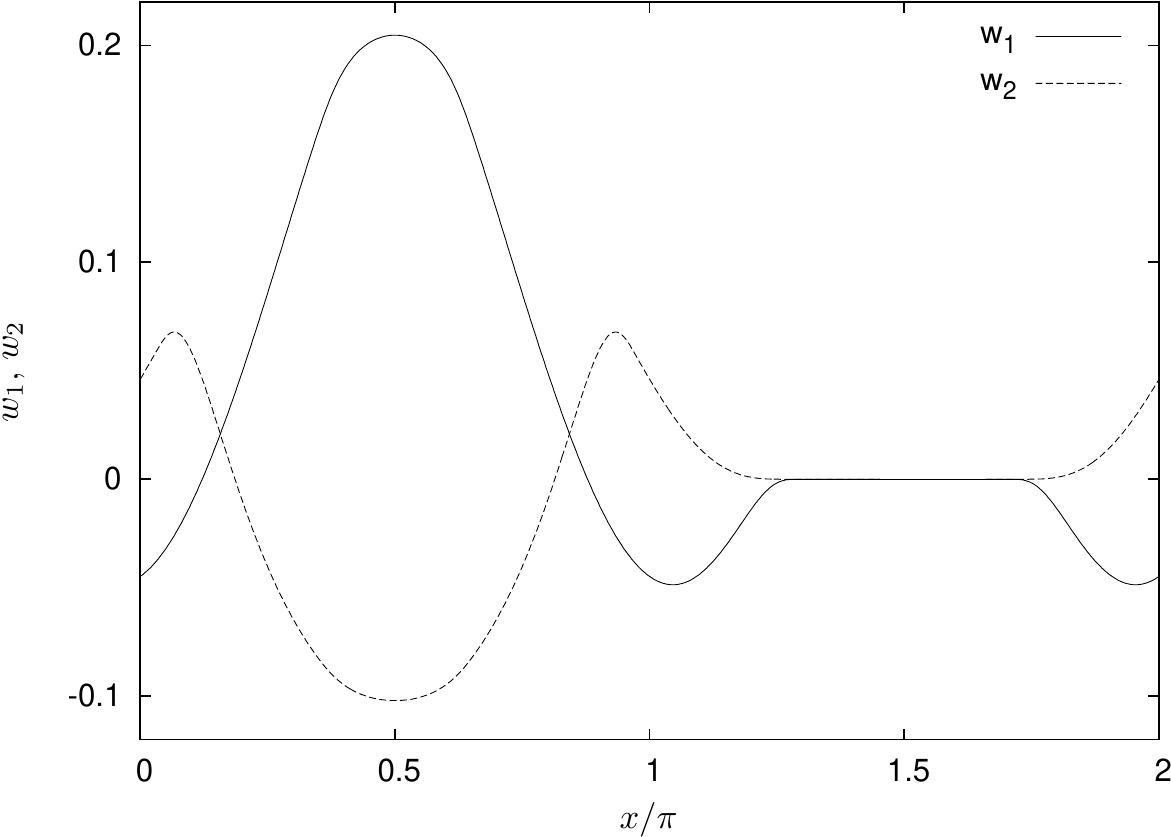}}
 \caption{Numerical solutions corresponding to asymptotic data
   in \Sectionref{sec:GowdyCauchyhorizon}.}
 \label{fig:CH1Dplots}
\end{figure}
In \Figref{fig:CH1Dplots}, we show the numerical
solution obtained from $N=1000$, $\Delta\tau=0.0025$ and $\tau_0=-18$.
We plot the Kretschmann scalar at two times $\tau=-10$ and
$\tau=0$. Hence, near the time $t=0$ (corresponding to
$\tau=-\infty$),
 the Kretschmann scalar is large on the
spatial interval $(0,\pi)$ while it stays bounded at $(\pi,2\pi)$. At
the later time, the curvature becomes smaller as expected. We also
plot the remainders $w^{(1)}$ and $w^{(2)}$ of $P$ and $Q$,
respectively. It is instructive to study how the polarized region inside
$(\pi,2\pi)$ gets ``displaced'' by the non-polarized solution.  


\section{Concluding remarks}

This paper presented a new approach to the singular initial value
problem for second-order Fuchsian-type equations. Our original motivation was to find a
reliable and accurate numerical scheme which, following \cite{ABL}, allowed us to evolve data 
numerically from the singularity. It turned out that the classical 
Fuchsian theory --despite its major
successes otherwise-- was not directly applicable to our puspose. 
This led us, on one hand, 
to develop a new approximation scheme, 
which is particularly natural to handle hyperbolic Fuchsian equations,
and, on the other hand, 
to revisit the classical theory and deduce a direct existence proof for
the singular initial value problem. 
Our scheme yields a reliable
and accurate numerical method, referred to here as
the {\sl Fuchsian numerical algorithm.}  
Importantly, we demonstrated that our method applies to
Gowdy-symmetric solutions to Einstein's field equations.

Our method should allow us to contribute to the understanding of
strong gravitational fields. In this direction, a particularly
interesting and outstanding problem in general relativity is the
strong cosmic censorship conjecture.  Our approach allows to
numerically construct, in particular, exceptional spacetimes, for
example solutions with Cauchy horizons.  This is of interest for two
reasons, at least. First, we can learn more about the ``solution
space'' of Einstein's field equations and, consequently, about the
validity of general relativity as a physical theory, due to the
unusual and sometimes physically ``undesired'' properties of these
exceptional solutions. We are currently investigating the geometries
of the solutions obtained in \Sectionref{sec:GowdyCauchyhorizon} in
greater detail \cite{BeyerLeFloch3}.


Second, our method allows one to study the stability of such
exceptional solutions, which was checked here,   
as a first step, by perturbing the induced data on a
spacelike hypersurface and computing the corresponding solution via 
the standard (backward) approach. This allows a systematic numerical
study of the strong cosmic censorship conjecture. Note, however, that
the strong cosmic censorship conjecture is known 
for Gowdy-symmetric solutions \cite{Ringstrom6,Ringstrom7}. Hence, in order to obtain new
interesting results about this conjecture we need to apply our theory
to more general classes of solutions; see \cite{BeyerLeFloch3}.

Our Fuchsian heuristics enables us
to distinguish between ``dominant'' and ``negligible'' terms in the equations as one approaches the ``singularity''. Sometimes these terms can be interpreted as physically interesting quantities 
like ``kinetic'' or ``potential energy''. Importantly, 
we discovered in the present paper 
that the 
principal part of the partial different system 
need not, by itself, 
determines the singular behavior of the solutions.  
Instead, nonlinear terms (classically treated as lower-order source-terms)  
often also play an important
role. This is so for Gowdy-symmetric  solutions, but it is even more
important for general solutions where mixmaster behavior is expected
according to the BKL conjecture. According to this
conjecture and more recent investigations \cite{Lim,Lim2}, spatial
derivative terms are expected to be insignificant except for exceptional points where spikes occur.


\section*{Acknowledgements}

The authors were partially supported by the Agence Nationale de la Recherche
(ANR) through the grant 06-2-134423 entitled {\sl ``Mathematical Methods in General Relativity''} (MATH-GR).
This paper was written during the year 2008--2009
when the first author (F.B.) was an ANR postdoctoral fellow at the Laboratoire J.-L. Lions.   



\begin{thebibliography}{1}

\small 

 \bibitem{ABL} 
 \auth{Amorim P., Bernardi C., and LeFloch P.G.,} 
 Computing Gowdy spacetimes via spectral evolution in future and past directions,  
 \jou{Class. Quantum Grav.} 26 (2009), 1--18.

\bibitem{Andersson}
 \auth{Andersson L.,}  
 The global existence problem in general relativity,
in: ``The Einstein equations and the large scale behavior of gravitational fields. 
50 Years of the Cauchy problem in general relativity'', Birkh\"auser, 2004, pp.~71--120. 

\bibitem{Andersson2} \auth{Andersson L., van Elst H., and Uggla C.},
  Gowdy phenomenology in scale-invariant variables,
  \jou{Class. Quantum Grav.} 21 (2004), S29--S57.  

\bibitem{Andersson3} \auth{Andersson L., van Elst H., Lim W.C., and Uggla
    C.}, Asymptotic silence of generic cosmological singularities,
  \jou{Phys. Rev. Lett.} 94 (2005), 051101.

\bibitem{Berger97} \auth{Berger B.K., Garfinkle D.}, Phenomenology of
  the Gowdy Universe on $T^3 \times R$, \jou{Phys. Rev. D} 57 (1998), 4767--4777.

\bibitem{BergerMoncrief} 
\auth{Berger B.K. and Moncrief V.,}
 Numerical investigations of cosmological singularities, 
 \jou{Phys. Rev. D} 48 (1993), 4676.

\bibitem{BeyerLeFloch1} \auth{Beyer F. and LeFloch P.G.,}
Second-order hyperbolic Fuchsian systems. General theory, 
ArXiv:1004.4885. 

\bibitem{BeyerLeFloch2} \auth{Beyer F. and LeFloch P.G.,}
Second-order hyperbolic Fuchsian systems. 
Gowdy spacetimes and the Fuchsian numerical algorithm, 
ArXiv:1006.2525.

\bibitem{BeyerLeFloch3} \auth{Beyer F. and LeFloch P.G.,}
in preparation.

\bibitem{Choquet52} \auth{Choquet-Bruhat Y.,}
Th\'eor\`eme d'existence pour certains syst\`emes d'\'equations aux d\'eriv\'ees partielles non lin\'eaires,
\jou{Acta Math.} 88 (1952), 141--225. 

\bibitem{Choquet69} \auth{Choquet-Bruhat Y. and and Geroch  R.,}
Global aspects of the {C}auchy problem in general relativity,
\jou{Comm. Math. Phys.} 14 (1969), 329--335. 

\bibitem{ChoquetIsenberg} \auth{Choquet-Bruhat Y. and Isenberg J.,}
Half-polarized ${\rm U}(1)$ symmetric vacuum spacetimes with AVTD behavior,
\jou{J. Geom. Phys.} 56 (2006), 1199--1214. 

 \bibitem{Choquet08} \auth{Choquet-Bruhat Y.,}
 Fuchsian partial differential equations, ``WASCOM 2007''---
 ``14th Conference on waves and stability in continuous media'', 
 World Sci. Publ., Hackensack, NJ, 2008, pp.~153--161. 

 \bibitem{Choquet09} \auth{Choquet-Bruhat Y.,}
 {\sl General relativity and the Einstein equations,}
 Oxford Math. Monographs, Oxford Univ. Press, Oxford, 2009.

\bibitem{Chrusciel} 
\auth{Chru\'sciel P.,} 
On spacetimes with $U(1) \times U(1)$ symmetric compact Cauchy surfaces,
\jou{Ann. Phys.} 202 (1990), 100--150.
 
 \bibitem{ChruscielIsenbergMoncrief} 
 \auth{Chru\'sciel P., Isenberg J., and Moncrief V.,} 
 Strong cosmic censorship in polarized Gowdy spacetimes, 
 \jou{Class. Quantum Grav.} 7 (1990), 1671--1680. 

\bibitem{ChruscielIsenberg} 
\auth{Chru\'sciel P. and Isenberg J.,} 
Nonisometric vacuum extensions of vacuum maximal globally hyperbolic spacetimes, 
\jou{Phys. Rev. D} 48 (1993), 1616--1628.

\bibitem{ChruscielLake} 
\auth{Chru\'sciel P. and Lake K.,} 
Cauchy horizons in Gowdy spacetimes,
\jou{Class. Quantum Grav.} 21 (2004), S153--S169.

\bibitem{Damour} \auth{Damour T., Henneaux M., and Nicolai H.}
Cosmological billards, \jou{Class. Quant. Grav.} 20 (2003), R145--R200. 


\bibitem{Eardley71} \auth{Eardley D., Liang E., and Sachs R.},
  Velocity-dominated singularities in irrotational dust cosmologies,
  \jou{J. Math. Phys.} 13 (1972), 99--106.

\bibitem{Geroch1}
\auth{Geroch R.}, A method for generating solutions of Einstein's equations,
\jou{J. Math. Phys.} 12 (1971), 918.

\bibitem{Geroch2} \auth{Geroch R.}, A method for generating new
  solutions of Einstein's equation. 2, 
  \jou{J. Math. Phys.} 13 (1972), 394--404.

\bibitem{Gowdy73}
 \auth{Gowdy R.H.,} 
 Vacuum space-times with two parameter spacelike isometry groups and
 compact invariant hypersurfaces: Topologies and boundary conditions,
 \jou{Ann. Phys.} 83 (1974), 203--241.

\bibitem{hawking} \auth{Hawking S.W. and Ellis G.F.R.,}
 {\sl The large scale structure of space-time,}
 Cambridge Univ. Press, Cambridge, 1973.

\bibitem{Heinzle}
\auth{Heinzle J.M., Uggla C., and R\"ohr N.}, The cosmological billard attractor,
\jou{Adv. Theor. Math. Phys.} 13 (2009), 293-407.


\bibitem{HennigAnsorg} 
\auth{Hennig J. and Ansorg M.,} 
Regularity of Cauchy horizons in $S^2 \times S^1$ Gowdy spacetimes,
\jou{Class. Quantum Grav.} 27 (2010), 065010.

 \bibitem{IsenbergMoncrief} 
 \auth{Isenberg J. and Moncrief V.,}
 Asymptotic behavior of the gravitational field and the nature of singularities in Gowdy spacetimes, 
 \jou{Ann. Phys.} 99 (1990), 84--122.

 \bibitem{KichenassamyRendall}
 Kichenassamy S. and Rendall A.D., 
 Analytic description of singularities in Gowdy spacetimes,
 \jou{Class. Quantum Grav.} 15 (1998), 1339--1355.

\bibitem{Kichenassamy} \auth{Kichenassamy S.,}
 {\sl Fuchsian reduction. Applications to geometry, cosmology and mathematical physics,}
 Birkh\"auser, Boston, 2007.

\bibitem{Kreiss2}
\auth{Gustafsson B., Kreiss H.O., and Oliger J.},
{\sl Time dependent problems and difference methods,} 
Wiley-Interscience, New York, 1995.

\bibitem{Kreiss} 
  \auth{Kreiss H.O., Petersson N.A., and Ystrom J.,}
  Difference approximations for the second-order wave equation, 
  \jou{Siam J. Numer. Anal} 40 (2002), 1940--1967.

\bibitem{Moncrief} 
 \auth{Moncrief V.,}
  Global properties of Gowdy spacetimes with $T^3 \times \RR$ topology, 
  \jou{Ann. Phys.} 132 (1981), 87--107.

\bibitem{LeVeque} 
\auth{LeVeque R.J.,}
{\sl Finite difference methods for ordinary and partial differential equations. 
Steady-state and time-dependent problems,}
Soc. Indust. Applied Math. (SIAM), Philadelphia, PA, 2007. 

\bibitem{Lim}
\auth{Lim W.C.}, The dynamics of inhomogeneous cosmologies, 
\jou{Ph.~D. thesis,} 
University of Waterloo (Canada), 2004. See arXiv:gr-qc/0410126.

\bibitem{Lim2} \auth{Lim W.C.}, 
New explicit spike solutions.
  Non-local component of the generalized mixmaster attractor,
  \jou{Class. Quantum Grav.} 25 (2008), 045014.

 \bibitem{Rendall00}
 \auth{Rendall A.D.,}
 Fuchsian analysis of singularities in Gowdy spacetimes beyond analyticity,
 \jou{Class. Quantum Grav.} 17 (2000), 3305--3316.

\bibitem{RendallWeaver} 
\auth{Rendall A.D. and Weaver M.,}
Manufacture of Gowdy spacetimes with spikes,
\jou{Class. Quantum Grav.} 18 (2001), 2959--2975. 


\bibitem{Ringstrom4} 
\auth{Ringstr\"om H.,}  
Curvature blow up on a dense subset of the singularity in T3-Gowdy,
\jou{J. Hyper. Diff. Equa.} 2 (2005), 547--564. 

 \bibitem{Ringstrom6} 
 \auth{Ringstr\"om H.,}  
 Strong cosmic censorship in $T^3$-Gowdy spacetimes, 
 \jou{Ann. of Math.} 170 (2009), 1181--1240.

 \bibitem{Ringstrom7} 
 \auth{Ringstr\"om H.}, Cosmic censorship for Gowdy spacetimes,
\jou{Living Reviews in Relativity,}
http://rela\-tivity.living\-reviews.org\-/articles/lrr-2010-2. 

\bibitem{Uggla}
\auth{Uggla C., van Elst H., Wainwright J., and Ellis G.F.R}, 
The past attractor in
inhomogeneous cosmology, \jou{Phys. Rev. D} 68 (2003), 103502.

\bibitem{Wainwright} \auth{Wainwright J. and Ellis G.F.R.,}
 {\sl Dynamical systems in cosmology,}
 Cambridge Univ. Press, Cambridge, 1997.

\end{thebibliography}
\end{document}